\documentclass[11pt,a4paper, reqno]{article}
\usepackage{fullpage}

\usepackage{amsfonts,amsmath,amsthm,amssymb,ifsym,cancel,appendix,enumerate,url,bbm,dsfont}
\usepackage{MnSymbol}
\usepackage{mathrsfs}
\usepackage[raggedright]{sidecap}
 \usepackage{tikz}

\newcommand{\supp}{\text{supp}\hspace{.06cm}}

\newcommand{\divergence}{\text{div}\hspace{.06cm}}
\newcommand{\grad}{\text{grad}\hspace{.06cm}}
\newcommand{\vol}{\text{vol}}

\theoremstyle{definition}

\theoremstyle{remark}
\newtheorem*{rem}{Remark}

\theoremstyle{plain}
\newtheorem{thm}{Theorem}[section]
\newtheorem*{thm*}{Theorem}
\newtheorem{lem}[thm]{Lemma}
\newtheorem{prop}[thm]{Proposition}
\newtheorem{cor}[thm]{Corollary}

\newtheorem{defn}[thm]{Definition}
\numberwithin{equation}{section}

\begin{document}

\title{\textbf{Klein-Gordon Solutions on Non-Globally Hyperbolic Standard Static Spacetimes}}

\author{David Bullock\footnote{Department of Mathematics, University of York, Heslington, York, YO10 5DD, UK. E-mail: \texttt{dmab500@york.ac.uk}}}

\date{\today}
\maketitle
\begin{abstract}
We construct a class of solutions to the Cauchy problem of the Klein-Gordon equation on any standard static spacetime. Specifically, we have constructed solutions to the Cauchy problem based on any self-adjoint extension (satisfying a technical condition: ``acceptability") of (some variant of) the Laplace-Beltrami operator defined on test functions in an $L^2$-space of the static hypersurface. The proof of the existence of this construction completes and extends work originally done by Wald. Further results include the uniqueness of these solutions, their support properties, the construction of the space of solutions and the energy and symplectic form on this space, an analysis of certain symmetries on the space of solutions and of various examples of this method, including the construction of a non-bounded below acceptable self-adjoint extension generating the dynamics. 
\end{abstract}

%

\section{Introduction}  
The first purpose of this paper is to construct a class of solutions to the Cauchy problem of the Klein-Gordon equation on any (not necessarily globally hyperbolic) standard static spacetime $(M,g)=(\mathbb{R}\times\Sigma, V^2dt^2-h)$, where $(\Sigma,h)$ is a Riemannian manifold and $V$ is a smooth positive function on $\Sigma$ (Sanchez \cite{a}). A class of solutions was originally constructed by Wald~\cite{b}. His solutions were given in terms of some fixed positive self-adjoint extension (s.a.e.) of a particular symmetric linear operator on $L^2(\Sigma,V^{-1}d\vol_h)$. Our treatment of the existence of solutions differs from that of Wald in the following aspects:
\begin{enumerate}
\item Wald considered only positive s.a.e.s and so the linear operators $C(t,A_E)$ and $S(t,A_E)$ (defined in Section~\ref{candidatesolutions}) used to construct solutions were bounded. In this paper however we also consider ``acceptable" s.a.e.s (Definition~\ref{acceptablesae}). Incidentally, all bounded below s.a.e.s are acceptable. Under these conditions $C(t,A_E)$ and $S(t,A_E)$ may be unbounded linear operators so care is required with the domains. 

\item We point out that a more recent result on the extendibility of subsets of the spacetime to smooth spacelike Cauchy surfaces in globally hyperbolic spacetimes by Bernal and Sanchez~\cite{h} is needed to complete the proof on the existence of Wald solutions. 
\end{enumerate}

The second purpose of this paper is to prove various properties satisfied by the solutions and to analyse examples. Since many already known results are quoted in this work for completeness, we shall for clarity list the other main results of this paper:

\begin{enumerate}
\item \label{number1} We show in detail the properties satisfied by the solutions only implicit in the paper by Wald and Ishibashi \cite{c}. In that paper, they assumed certain conditions on the dynamics (e.g.\ constraints on the support of solutions, how they are transformed under time translation and reflection in time and the existence of an energy norm) and then proved that it must be generated by a particular s.a.e.\ $A_E$. In this paper we answer the natural question: ``to what extent are these conditions on the dynamics necessary?", that is, does the dynamics generated by a particular choice of acceptable s.a.e.\ $A_E$ satisfy these conditions? We shall answer mostly to the affirmative. However we note that Assumption~1 in Wald and Ishibashi \cite{c} (the support of solutions/``the causality assumption": see below) is not always true of dynamics generated by an arbitrary acceptable s.a.e.. 
\vspace{0em}

To amplify this point, Assumption 1 on the dynamics in Wald and Ishibashi \cite{c} states that the support of the solution to the Klein-Gordon equation corresponding to Cauchy data always lies within the union of the causal future and past of the support of that data. In Section~\ref{counterexample} we give a simple example of a standard static spacetime and a choice of s.a.e.\ $A_E$ such that the dynamics generated satisfies: $\supp \phi\nsubseteq J(K)$ for some initial data $(\phi_0,\dot{\phi}_0)$, where $K=\supp\phi_0\cup\supp\dot{\phi}_0$. We show in Section~\ref{sec:supportofwaldsolutions} however that, in general, $\supp \phi$ is contained in $J(K)$ up until the time at which the data can ``hit'' any edge in the spacetime. We prove that this weaker form of Assumption 1 is true of all dynamics constructed in this paper, using the previous results on the uniqueness of solutions in Section~\ref{sec:uniquenessofwaldsolutions} and results on the causal structure of the spacetime Section~\ref{sec:causalstructure (ii)}.

\item An important property satisfied by the ``Wald solutions'' is that the value of the standard symplectic form evaluated at any pair of solutions is independent of the static hypersurface on which it is calculated, so the space of solutions has a natural symplectic space structure. Since it is this structure which allows the quantisation of the theory, by the construction of the Weyl algebra (B\"{a}r et al.\ \cite{e}), it is an important result in Section~\ref{sec:sympform} that even after extending Wald's method to the case of only acceptable s.a.e.s, we retain the conservation of the symplectic form even in the cases where the positive definiteness of the energy form (Section~\ref{sec:bilinearform}) is lost.

\item In Section~\ref{symmetries} we prove how the solutions are transformed under time translation and reflection. We show that the previously constructed energy form is invariant under both time translation and reflection of its arguments whereas the symplectic form is time-translation invariant but acquires a minus sign under reflection of its arguments in time. (These properties correspond to assumptions 2(i), 2(ii), 3(i) and 3(ii) of Wald and Ishibashi \cite{c}.)

\item In Sections~\ref{circle}-\ref{interval}, as examples we consider three simple one-dimensional Riemannian manifolds ($S^1$, $(0,\infty)$ and $(0,a)$ with their usual differential structures and Riemannian metrics), each of which will then generate a standard static spacetime with $V=1$. In order to classify the dynamics generated on the latter spacetimes by the construction of this paper, we give the s.a.e.s of minus the Laplacian on $S^1$, $(0,\infty)$ and $(0,a)$ and determine their spectra and resolvents. The proofs of these statements are to be found in Appendices~F,G and H of Bullock~\cite{me}. 

\item In Section~\ref{acceptablenonbounded} we construct an acceptable non-bounded below s.a.e.\ $A_E$ of minus the Laplacian on $\Sigma=\mathbb{Z}\times(0,\infty)$. This example then shows that the extension of theory of Wald~\cite{b} from bounded-below s.a.e.s to acceptable s.a.e.s carried out in this paper is non-trivial (Wald's paper only deals  with positive s.a.e.s).
\end{enumerate}

\section{The Cauchy problem of the Klein-Gordon Equation on Standard Static Spacetimes: The construction of candidate solutions as vector-valued functions}\label{candidatesolutions}

We start by defining the class of spacetimes, with which this paper is concerned (see e.g. Sachs \& Wu \cite{l}, Sanchez \cite{a} and O'Neill \cite{g}).

\begin{defn}
(Sanchez \cite{a})\label{standardstatic}  A \textbf{standard static spacetime} is defined by: $$(M,g)=(\mathbb{R}\times \Sigma,V^2 dt^2-h),$$ where $(\Sigma,h)$ is a smooth Riemannian manifold; $M$ is given the usual product topology and differential structure; $dt^2$ is the Euclidean metric on $\mathbb{R}$; $V\in C^{\infty}(\Sigma)$ with $V>0$. The time-orientation is that given by the timelike vector field $X=\frac{\partial}{\partial t}$.
\end{defn}
\begin{rem}
In the above expression for the metric $g=V^2 dt^2-h$, we are using a slightly sloppy notation for conciseness. Denoting by $\pi_1: \mathbb{R}\times \Sigma\rightarrow \mathbb{R}$ and $\pi_1: \mathbb{R}\times \Sigma\rightarrow\Sigma$ the two projection (bundle) maps, then more precisely: $g=\pi_2^{*}(V^2)\pi_1^*( dt^2)-\pi_2^{*}(h)$, where $\pi_i^*$ is the pull-back applied here to metrics and functions and $dt^2:=dt\otimes dt$ is the standard Riemannian metric on $\mathbb{R}$. 
\end{rem}

Let $(M,g)=(\mathbb{R}\times \Sigma,V^2 dt^2-h)$ be a standard static spacetime. As $\Sigma$ is a smooth manifold then for each $t\in\mathbb{R}$ the map $\pi_t:\Sigma\rightarrow M$ given by $x\rightarrow (t,x)$ is a smooth embedding from $\Sigma$ to $M$ and for each embedded submanifold $\Sigma_t:=\{t\}\times\Sigma=\pi_t(\Sigma)\subset M$ there exists a unique unit future-pointing smooth timelike vector field $n_t=V^{-1}\frac{\partial}{\partial t}$ normal to each tangent space of $\Sigma_t$. Note that we have not assumed that the manifold $\Sigma$ is orientable.

We wish to solve the Klein-Gordon equation on an arbitrary standard static spacetime. For an arbitrary spacetime and mass  $m\geq 0$ the Klein-Gordon equation reads:
\begin{equation}\label{E:5}
(\largesquare_g +m^2) \phi = 0,
\end{equation}
where $\largesquare_g=\divergence_g \circ \grad_g$ is the Laplace-Beltrami operator (see Appendix D of Bullock~\cite{me}), sometimes locally given by: $\nabla^{\mu}\nabla_{\mu}$ where $\nabla_{\mu}$ is the covariant derivative defined by the metric. Alternatively, the Klein-Gordon equation can be expressed in local coordinates: $$\largesquare_g=\frac{1}{\sqrt{|g|}}\partial_{\mu}g^{\mu\nu}\sqrt{|g|}\;\partial_{\nu}$$ 
where $g:=\det(g_{\mu\nu})$.
\begin{rem} Note also that we shall demand that $\phi\in C^\infty(M)$, where $C^\infty(M)$ is defined as the space of all smooth $\mathbb{K}$-valued functions on $M$, where $\mathbb{K}=\mathbb{R}$ or $\mathbb{C}$. We are removing the dependence of the field of scalars from our notation for $C^\infty(M)$ purely for brevity. We shall find that the results of this paper apply equally well to solving the Klein-Gordon equation for real-valued functions as for complex-valued functions. In the sequel we shall take all function spaces, Hilbert spaces etc.\ to be either over $\mathbb{R}$ or $\mathbb{C}$ as required. We shall on occasion in this paper mention where we may have to treat the two cases separately. For instance Sections~\ref{circle}-\ref{interval} only apply to the complex case as will be discussed there.   
\end{rem}

Our spacetime of interest is: $M=\mathbb{R}\times\Sigma$ with $g=V^2dt^2-h$, where $\Sigma$ is a smooth manifold with smooth Riemannian metric $h$ and $V\in C^\infty(\Sigma)$, $V>0$. For this spacetime, we define a solution to the Cauchy problem for the Klein Gordon equation to be a linear map: \begin{align*}
\Psi: C_0^\infty(\Sigma_0)\times C_0^\infty(\Sigma_0) &\rightarrow C^\infty (M)\\
(\phi_0,\dot{\phi}_0)&\mapsto \phi,
\end{align*}
such that, for all $\phi_0,\dot{\phi}_0\in C_0^\infty(\Sigma_0)$, if $\Psi(\phi_0,\dot{\phi}_0)=\phi$ then:
\begin{enumerate} 
\item$(\largesquare_g +m^2)\phi=0$
\item $\phi|_{\Sigma_0}=\phi_0$
\item $\partial_t\phi|_{\Sigma_0}=\dot{\phi}_0$.
\end{enumerate}
In this paper, we shall construct solutions to the Cauchy problem. (We shall in fact find a solution to an extension of this problem, that is, extend the space of test functions $C_0^\infty(\Sigma_0)$ to a certain subspace $\chi_E$ of $C^\infty(\Sigma_0)$.) We start by expressing the Klein-Gordon equation in a simpler form. 
Given an atlas $(U_\alpha,\phi_\alpha)$ for $\Sigma$ we have the following atlas for $M$: $(\mathbb{R}\times U_\alpha,t\times \phi_\alpha)$.
In these local coordinates the Laplace-Beltrami operator reads: 
$$\largesquare_g=V^{-2}\partial^2_t-V^{-1}D^iVD_i,$$
where $D_i$ is the covariant derivative on $\Sigma$ induced by $h$. Thus equation \eqref{E:5} reads:
$(V^{-2}\partial^2_t-V^{-1}D^iVD_i+m^2)\phi=0$ iff $(\partial^2_t-VD^iVD_i+m^2V^2)\phi=0$ iff:
\begin{equation}\label{kleingordonsimp}
\partial^2_t\phi=-A\phi,
\end{equation}
where $A=-VD^iVD_i+m^2V^2$. Note that in coordinate free notation: $A=-V\divergence_h V\grad_h +m^2V^2$. See Appendix~D in Bullock~\cite{me} for more details.  

We solve this form of the Klein-Gordon equation with the methods of functional analysis on the (real or complex) Hilbert space $L^2(\Sigma,V^{-1}d\vol_h)$. (This space is defined in e.g.\ Section~D.3 of Bullock~\cite{me}. 
On the (real or complex) Hilbert space $L^2 (\Sigma,V^{-1}d\vol_h)$, we have the following linear operator $A$:  \begin{align}\label{operatorA}
D(A)&=[C_0 ^{\infty}(\Sigma)]\\
A[\phi]&=[(-VD^iVD_i+m^2V^2)\phi],
\end{align}
where $D_i$ is the covariant derivative on $(\Sigma,h)$ and $\phi\in C_0^\infty(\Sigma)$. That $A$ is symmetric and  positive is proven in Proposition~D.10 in the Bullock~\cite{me}. Note that the adjoint is a well-defined linear operator since $A$ is densely defined. Its adjoint $A^*$ is given by \label{adjointofA}: 
\begin{align*}
D(A^*)&=\{\phi\in L^2(\Sigma, V^{-1}d\vol_h)\text{ s.t. } A\phi\in L^2(\Sigma, V^{-1}d\vol_h)\}\\
A^*\phi&=A\phi,
\end{align*}
where in both lines $A\phi$ is meant distributionally and a priori $\phi,A\phi\in D'(\Sigma)$. Here, functions are interpreted as distributions by use of the smooth measure $V^{-1}d\vol_h$ on $\Sigma$. 

\begin{rem}
Note that since $[C_0^\infty(\Sigma)]\subseteq D(A^*)$ then $D(A^*)$ is densely defined and so $A$ is closable. Also, be aware that the reason for the appearance of the partial differential operator $A$ instead of its formal adjoint $A^*$ in the above definition of the linear operator $A^*$ is that $A$ is formally self-adjoint with respect to the smooth measure $V^{-1}d\vol_h$. See Section~D in Bullock~\cite{me} for definitions of these terms. 
\end{rem}

The domain of the closure $\overline{A}$ of $A$ is given by the closure of $[C_0^\infty(\Sigma)]$ in the Hilbert space $D(A^*)$ with the inner product $\langle\cdot,\cdot\rangle_{A^*}$:
$$\langle \phi,\theta\rangle_{A^*}= \langle \phi,\theta\rangle_{L^2(\Sigma, V^{-1}d\vol_h)}+\langle A^*\phi,A^*\theta\rangle_{L^2(\Sigma, V^{-1}d\vol_h)},$$

It is important to note that $A$ is not necessarily essentially self-adjoint (e.s.a.). The following well-known theorem gives a case where $A$ is e.s.a.. 

\begin{thm}[Essential Self-Adjointness of minus the Laplacian on Complete Riemannian Manifolds]\label{esa}
Let $(\Sigma,h)$ be a complete Riemannian manifold. Then letting $V=1$ and $m=0$, we have $A=-\divergence_h \grad_h=-\Delta_h$, minus the Laplacian corresponding to the metric $h$. Then if $D(A)=[C^\infty_0(\Sigma)]$ in the Hilbert space $H=L^2(\Sigma, d\vol_h)$, then $A$ is essentially self-adjoint.
\end{thm}
\begin{proof}
See e.g.\ Taylor \cite{taylor2}, Proposition 8.2.4.
\end{proof}

As pointed out by Wald \cite{b}, since $A$ is a symmetric positive linear operator then at least one positive self-adjoint extension exists. We do not restrict ourselves however to using a single extension, but we are forced to only consider a certain class of s.a.e.s of $A$, which we define shortly. We wish to first make a remark concerning the choice of the field of scalars. 
\begin{rem}
Note that if we define $H_\mathbb{K}=L^2(\Sigma,\mathbb{K},V^{-1}d\vol_h)$ as the space of equivalence classes of $\mathbb{K}$-valued square-integrable Borel-measurable functions, where $\mathbb{K}= \mathbb{R}$ or $\mathbb{C}$ and $f\sim g$ iff $f=g$ a.e., then we can view $A$ as a symmetric linear operator on either the real Hilbert space $H_\mathbb{R}$ or the complex Hilbert space $H_\mathbb{C}$. The set of self-adjoint extensions of these operators are related. To see how, take the general situation of a real Hilbert space $H$ and its complexification $H_\mathbb{C}$. Now, on $H_\mathbb{C}$ can be defined a natural complex conjugation operator $C$. It is shown in Section 2 of Seggev \cite{o} that the self-adjoint extensions of a symmetric linear operator $A$ on $H$ are are in bijection with the self-adjoint extensions of the symmetric linear operator $A_\mathbb{C}$ on $H_\mathbb{C}$, which commute with $C$, where $A_\mathbb{C}$ is the complexification of $A$.
\end{rem} 
We now introduce our new notion of an acceptable s.a.e.: 
\begin{defn}\label{acceptablesae}
A s.a.e.\ $A_E$ of $A$ is called \textbf{acceptable} if it satisfies:
\begin{equation}
[C_0^\infty(\Sigma)]\subseteq \bigcap_{t>0}D(\exp(A_E^-)^{1/2}t),
\end{equation}
where $A_E^-:=x^-(A_E)$ is the positive self-adjoint operator defined via continuous functional calculus using the function $x^-:\mathbb{R}\rightarrow [0,\infty)$ defined by:
$$x^-(y):=\left\{\begin{array}{ll}-y,&y\leq 0\\
0,&\text{otherwise.}\end{array}\right.$$ The operator $A_E^-$ is called the negative part of the operator $A_E$ and it's bounded iff $A_E$ is bounded-below (i.e. there exists $M\in\mathbb{R}$ such that $\langle Ax|x\rangle\geq -M||x||^2$ \text{ for all } $x\in D(A)$ iff $\sigma(A)\subseteq [-M,\infty)$.
\end{defn}

\begin{rem}
In the paper \cite{b} by Wald, he considered only positive s.a.e.s of $A$. Clearly, a positive linear operator is bounded-below. If $A_E$ is a bounded-below s.a.e.\ then $A_E^-$ is a bounded linear operator, as is $(A_E^-)^{1/2}$. Then $\exp(A_E^-)^{1/2}t$ is also a bounded linear operator for all $t$ and so: $$[C_0^\infty(\Sigma)]\subseteq L^2(\Sigma,V^{-1}d\vol_h)=\bigcap_{t>0}D(\exp(A_E^-)^{1/2}t).$$ Thus every bounded-below s.a.e $A_E$ is also acceptable. Thus we are extending the method of Wald to more s.a.e.s of $A$. 
\end{rem}

The approach (taken from Wald \cite{b}) is to find a map $\mathbb{R}\rightarrow D(A_E)\subseteq H$, where $H=L^2(\Sigma,V^{-1}d\vol_h)$. $t\rightarrow \phi_t$, for each pair of data $\phi_0,\dot{\phi}_0\in C^\infty_0(\Sigma)$. We demand that the map $t\rightarrow\phi(t)$ is twice differentiable as a vector-valued function with double-derivative
\begin{equation}\label{stuff}
\frac{d^2 \phi_t}{dt^2}=-A_E \phi_t.
\end{equation}

Our intended solution to this problem is given in terms of any acceptable s.a.e.\ $A_E$ of $A$:
\begin{equation}\label{phit}
[\phi_t]=\cos(A_E ^{1/2}t)[\phi_0] +A_E ^{-1/2}\sin(A_E ^{1/2}t)[\dot{\phi}_0]
\end{equation}

Our immediate problem is to show that this expression makes sense. If $A_E$ was positive self-adjoint then, following Wald \cite{b}, we can take the square root to form a positive self-adjoint unbounded linear operator $A_E^{1/2}$ and then construct the two bounded linear operators $\cos(A_E ^{1/2}t)$ and $A_E ^{-1/2}\sin(A_E ^{1/2}t)$ by applying the multiplication operator form of the Spectral Theorem, as in Reed and Simon \cite{f}. If $A_E$ was not positive but merely bounded-below, then we shall show that this method still works and $\cos(A_E ^{1/2}t)$ and $A_E ^{-1/2}\sin(A_E ^{1/2}t)$ are still well-defined bounded linear operators despite the non-existence of the square root. If, however $A_E$ is not bounded-below then these linear operators will be unbounded and we must concern ourselves with their (dense) domains. We shall show that even when $A_E$ is not bounded-below, but is acceptable (Definition~\ref{acceptablesae}), then we can solve the Cauchy problem with respect to smooth initial data of compact support.

In this paper, in order to avoid expressions involving square roots of non-positive self-adjoint linear operators we introduce an alternative representation of equation \eqref{phit}. Define the functions $C,S:\mathbb{R}^2\rightarrow \mathbb{R}$:
 \begin{align*}
 C(t,x)& =\cos(x^{1/2}t)=\left\{\begin{array}{ll}
\cos(x^{1/2}t) &\text{ for }x\geq 0\\
\cosh((-x)^{1/2}t)&\text{ for }x<0\end{array}\right. \\
 S(t,x)& =t\frac{\sin(x^{1/2}t)}{x^{1/2}t}=\left\{\begin{array}{ll}
x^{-1/2}\sin(x^{1/2}t)&\text{ for }x\geq 0\\
(-x)^{-1/2}\sinh((-x)^{1/2}t)&\text{ for }x<0\end{array}\right.,
 \end{align*}
 where: $$\frac{\sin z}{z}:=\sum_{n=0}^\infty\frac{(-1)^n z^{2n}}{(2n+1)!}\text{ and }\frac{\sinh z}{z}:=\sum_{n=0}^\infty\frac{z^{2n}}{(2n+1)!}$$ are analytic functions on $\mathbb{C}$, both being invariant under $z\rightarrow -z$ (the same is true of course of the functions $\cos z$ and $\cosh z$). This makes the definitions of $C(t,x)$ and $S(t,x)$ independent of the choice of square root. 
Since $C(t,\cdot)$ and $S(t,\cdot)$ are (unbounded) real-valued measurable functions for each fixed $t$, then by functional calculus we can construct the (possibly unbounded) self-adjoint linear operators $C(t,A_E)$ and $S(t,A_E)$, for any s.a.e.\ $A_E$ of $A$. It is shown in the remarks following Propositions~A.7 and A.8 of Bullock~\cite{me} that for $t\geq 0$: $$D(\exp(A_E^-)^{1/2}t)=D(C(t,A_E))\subseteq D(S(t,A_E)).$$
Thus:
$$[C_0^\infty(\Sigma)]\subseteq \bigcap_{t>0}D(\exp(A_E^-)^{1/2}t)=\bigcap_{t>0}D(C(t,A_E))$$
and the condition on $A_E$ that it is acceptable is precisely what is required for $C(t,A_E)$ (and so $S(t,A_E)$) to be defined on equivalence classes of test functions.

Given an acceptable s.a.e.\ $A_E$ of $A$, let our proposed solution to Equation~$(\ref{stuff})$ for arbitrary $\phi_0,\dot{\phi}_0\in C^\infty_0(\Sigma)$ define:
\begin{equation}\label{stuff2}
[\phi_t]=C(t, A_E)[\phi_0]+S(t,A_E)[\dot{\phi}_0].
\end{equation}
If $A_E$ is bounded-below, then $C(t,A_E)$ and $S(t,A_E)$ are bounded linear operators for all $t$ (proven in Appendix~A of Bullock~\cite{me}) and $[\phi_t]$ is a well-defined element of $L^2(\Sigma, V^{-1}d\vol_h)$. If not, then the condition on $A_E$ in Definition~\ref{acceptablesae} is precisely what is required for the RHS to make sense. We wish to show that in fact the map $t\rightarrow [\phi_t]$ is infinitely differentiable and that $[\phi_t]\in [C^\infty(\Sigma)]\cap L^2(\Sigma, V^{-1}d\vol_h)$ for all $t\in \mathbb{R}$.

The following proposition is vital for this paper. It is an application of Sobolev theory. It is taken from Wald~\cite{b} and reproduced here for completeness. (For the definitions of $L^p$ spaces, distributions and Sobolev spaces $W^{k,p}(M,\mu)$ on a Riemannian manifold $M$ with smooth measure $\mu$, see Appendix~D.3 of Bullock~\cite{me}.)

\begin{thm}\label{sobolev}
Any s.a.e.\ $A_E$ of $A$ satisfies: $D(A_E^\infty)\subseteq [C^\infty(\Sigma)]$.
\end{thm}

\begin{proof}({Wald \cite{b}})
We know that $[C_0^\infty(\Sigma)]=D(A)\subseteq D(A_E^\infty)$. 
Take $\phi\in D(A_E^\infty)$. Since $L^2(\Sigma, V^{-1}d\vol_h)\subseteq L^1_{loc}(\Sigma, V^{-1}d\vol_h)\subseteq D'(\Sigma)$, the space of distributions on the manifold $\Sigma$, then for all $f\in C^\infty_0(\Sigma)$:
$$\phi(A^n f)=\langle \phi, [A^n f]\rangle=\langle \phi, A^n [f]\rangle=\langle \phi, A_E^n [f]\rangle=\langle A_E^n\phi, [f]\rangle=(A_E^n\phi)(f) $$
Thus $A^n\phi\in D(A_E^\infty)\subseteq L^2(\Sigma, V^{-1}d\vol_h)$, where $A^n\phi$ is interpreted in the sense of distributions (since $A$ is a formally self-adjoint partial differential operator of second order w.r.t.\ $V^{-1}d\vol_h$, see Appendix~D.2 of Bullock~\cite{me}).

Take an open set $\Omega\subseteq \Sigma$, which is precompact in the domain of a chart on $\Sigma$. Letting $N:=\dim\Sigma$, then denote the resulting chart map $\Psi:\Omega\rightarrow \mathbb{R}^N$. Restricting $\phi$ to $\Omega$, we have $$A^n\phi\in L^2(\Omega, V^{-1}d\vol_h)=W^{0,2}(\Omega,V^{-1}d\vol_h).$$ As $V^{-1}$ and $|\det(h_{ij})|$ are bounded by below on $\Omega$, then $A^n\phi\in W^{0,2}(\Psi(\Omega))\subseteq W^{0,2}_{loc}(\Psi(\Omega))$, where we are now viewing $\phi$ as a function and $A$ as a p.d.o.\ on $\Psi(\Omega)\subseteq\mathbb{R}^N$. As $A^n$ is an elliptic p.d.o.\ of order $2n$, then, by an elliptic regularity theorem (Theorem~D.12 of Bullock~\cite{me}), $\phi\in W^{2n,2}_{loc}(\Psi(\Omega))$ for all $n$. And by Sobolev's lemma (Theorem~D.13 of \cite{me}), we have (after possibly changing $\phi$ on a null set) that $\phi\in C^l(\Psi(\Omega))$ for any non-negative integer $l<2n-\frac{N}{2}$. Since $n$ and $\Omega$ are arbitrary, then $\phi\in C^\infty(\Sigma)$.
\end{proof}

Using Theorem~\ref{sobolev}, we define a space of smooth functions $\chi_E$ which contains all compactly supported smooth functions (as $A_E$ is acceptable). We shall show in later sections that we can solve the Klein-Gordon equation with respect to data in the space $\chi_E$.
\begin{prop}\label{strongderivs}
Given an acceptable s.a.e.\ $A_E$ of $A$, define:
$$\chi_E:=\{f\in C^\infty(\Sigma) \text{ s.t. }[f]\in D(A_E^\infty)\cap\bigcap_{t>0}D(\exp((A_E^-)^{1/2}t))\}$$
Then the linear operators $C(t,A_E)$ and $S(t,A_E)$ satisfy 
$$C(t,A_E) \text{, }S(t,A_E):[\chi_E]\rightarrow [\chi_E].$$
Also, the maps $t\rightarrow C(t,A_E)$ and $t\rightarrow S(t,A_E)$ are infinitely often strongly differentiable on $[\chi_E]$,
where for $n\in \mathbb{N}\cup \{\infty\}$: $$D(A_E^n)=\{x \in D(A_E):\;A_E^m x\in D(A_E) \text{ for all } m=1,...,n-1\}$$
For $n\in \mathbb{N}$ the following strong derivatives hold on the dense subspace $[\chi_E]$ of $L^2(\Sigma, V^{-1}d\vol_h)$:
$$\begin{array}{ll}
\frac{d^{2n}}{dt^{2n}}C(t,A_E)=(-1)^n A_E^n C(t,A_E),
&\frac{d^{2n-1}}{dt^{2n-1}}C(t,A_E)=(-1)^n A_E^n S(t,A_E)\\
\frac{d^{2n}}{dt^{2n}}S(t,A_E)=(-1)^n A_E^n S(t,A_E),
&\frac{d^{2n+1}}{dt^{2n+1}}S(t,A_E)=(-1)^n A_E^n C(t,A_E).
\end{array}$$
\end{prop}
\begin{proof}
See Appendix~A of Bullock~\cite{me}.
\end{proof}
\begin{lem}
\label{stuff3}
Thus just as in Equation~\eqref{stuff2}, given initial data $\phi_0,\dot{\phi}_0\in \chi_E$ and letting $[\phi_t]=C(t,A_E)[\phi_0] +S(t,A_E)[\dot{\phi}_0]$, then $[\phi_t]$ is differentiable as a vector valued function to arbitrary order and to even order:
$$\frac{d^{2n}}{dt^{2n}}[\phi_t]=(-1)^n A_E^n [\phi_t]$$
Thus in particular for $n=1$ we have reproduced equation \eqref{stuff} and:
\begin{align*}
[\phi_0]&=[\phi_t]|_{t=0}\\
[\dot{\phi}_0]&=\left.\frac{d}{dt}[\phi_t]\right|_{t=0}
\end{align*}
\end{lem}

Since $[\chi_E$] is an invariant subspace of $L^2(\Sigma,V^{-1}d\vol_h)$ w.r.t.\ the linear operators $C(t,A_E)$ and $S(t,A_E)$, so, for all initial data $\phi_0,\dot{\phi}_0\in \chi_E$, the solution given in Proposition $\ref{stuff3}$ satisfies: $$[\phi_t]\in D(A_E^\infty)\subseteq [C^{\infty}(\Sigma_t)]\;\;\forall t\in \mathbb{R}.$$

Thus we have solved the Hilbert space version of the Klein-Gordon Equation (\eqref{stuff}). We shall use this in Section~\ref{existence1} to construct solutions of the Klein-Gordon equation itself (Equation~\eqref{kleingordonsimp}).

\section{Causal Structure of Standard Static Spacetimes (i)}\label{sec:causalstructure (i)}
Before we construct solutions to the Klein-Gordon equation in Section~\ref{existence1}, we shall find it useful to introduce some concepts from geometry, namely we shall define the causality relations and define the causal future and causal past of a set. After some preliminaries concerning Riemannian manifolds we shall then analyse the causal structure of an arbitrary standard static spacetime. Later in the section, since we shall need to quote results concerning the well-posedness of the Klein-Gordon equation on globally-hyperbolic spacetimes when we construct our solutions in Section~\ref{existence1}, so we define the terms globally hyperbolic, Cauchy surfaces and Cauchy developments. Subsequently, returning to standard static spacetimes, we shall then in Proposition~\ref{DSigmaexplicit} re-express the Cauchy development $D(\Sigma_0)$ of the hypersurface $\Sigma_0$. Lastly, in Theorem~\ref{globhypsolution1} we shall quote the well-known result concerning the well-posedness of the Klein-Gordon equation on globally-hyperbolic spacetimes.

We shall shortly begin to analyse the causal structure of a standard static spacetime.. However, we shall find it useful to first discuss metrics on Riemannian manifolds (here we use the term ``metric'' as in ``metric space'' rather than as in ``metric tensor''!). It is well known that a Riemannian manifold $(\Sigma, h)$ is naturally metrisable. A metric $d\colon \Sigma\times\Sigma\rightarrow [0,\infty)$ is given by: 
$$d(p,q)=\inf \left\{\begin{array}{lr}
\int_a^b|\dot{\sigma}(t)|dt\text{ s.t.}&\sigma\colon[a,b]\rightarrow \Sigma \text{ is a piecewise smooth}\\  &\text{curve in }\Sigma \text{ with } \sigma(a)=p,\;\sigma(b)=q.
\end{array}\right\},$$
where $|\dot{\sigma}(t)|:=[h_{\sigma(t)}(\dot{\sigma}(t),\dot{\sigma}(t))]^{1/2}$.

\begin{thm}\label{Riemmetric}
Given a Riemannian manifold $(\Sigma, h)$, then the metric $d$ given above induces the topology on $\Sigma$.
\end{thm}
\begin{proof}
See for example Lee \cite{p}, Lemma 6.2.
\end{proof}

For a choice of standard static spacetime $(M,g)=(\mathbb{R}\times \Sigma, V^2dt^2-h)$, we shall always choose the metric on $\Sigma$ induced by the Riemannian metric $V^{-2}h$ on $\Sigma$. The importance of choosing a metric on $\Sigma$ dependent on $V$ shall be seen in Proposition~\ref{J(K)compact}.

\begin{prop}\label{causalfuture} Given the standard static spacetime $(M,g)=(\mathbb{R}\times \Sigma, dt^2-h)$ and $K\subseteq \Sigma_0$, then $(t,y)\in J^+(K)$ iff there exists a smooth curve $\sigma\colon[0,t]\rightarrow \Sigma$ s.t.:
$\sigma(0)=x\in \pi(K),\;\sigma(t)=y\text{ and }|\dot{\sigma}(s)|\leq 1 \;\;\forall s\in [0,t]$.
\end{prop}
\begin{proof}
Given such a curve, define the curve $\gamma\colon[0,t]\rightarrow M$, by: $\gamma(s)=(s,\sigma(s))$.
Clearly $\gamma$ is smooth. It is also causal since $g_{\gamma(s)}(\dot{\gamma}(s),\dot{\gamma}(s))=1-h_{\sigma(s)}(\dot{\sigma}(s),\dot{\sigma}(s))\geq 0$. From\linebreak $g_{\gamma(s)}\left(\dot{\gamma}(s),\left.\frac{\partial}{\partial t}\right |_{\gamma(s)}\right)=1>0$ it follows that $\gamma$ is future-pointing. Then $\gamma$ is a smooth future-pointing causal curve from $(0,x)$ to $(t,y)$ and thus $(t,y)\in J^+((0,x))$.

Conversely, if $(t,y)\in J^+(0,x)$ then there exists $\gamma\colon[a,b]\rightarrow \mathbb{R}\times \Sigma$ which is smooth future-pointing and causal s.t. $\gamma(a)=(0,x)$ and $\gamma(b)=(t,y)$. Let $\gamma(s)=(\gamma_1(s),\gamma_2(s))$, where $\gamma_1\colon[a,b]\rightarrow \mathbb{R}$ and $\gamma_2\colon[a,b]\rightarrow \Sigma$ are both smooth curves defined using the smooth projection maps. So, $g_{\gamma(s)}(\dot{\gamma}(s),\dot{\gamma}(s))=|\dot{\gamma}_1(s)|^2-h_{\gamma_2(s)}(\dot{\gamma}_2(s), \dot{\gamma}_2(s))\geq 0$. The condition of $\gamma$ being future-pointing gives us: $g_{\gamma(s)}(\dot{\gamma}(s), \left.\frac{\partial}{\partial t}\right |_{\gamma(s)})=\dot{\gamma}_1(s)>0$. We wish to reparametrise this curve and show that it is of the form of the previous proposition. For this purpose, let $\Phi\colon[a,b]\rightarrow \mathbb{R}$ be given by: $\Phi(s)=\int^s_a \dot{\gamma}_1(u)du$. Since $\dot{\Phi}(s)=\dot{\gamma}_1(s)>0$, then, by the Inverse Function Theorem, there exists a smooth inverse $\Phi^{-1}\colon[0,c]\rightarrow [a,b]$, where $\Phi(a)=0$, $\Phi(b)=c$. Now, define the reparametrisation: $\gamma'(s)=\gamma(\Phi^{-1}(s))$. $\gamma'\colon[0,c]\rightarrow \mathbb{R}\times \Sigma$ is then a smooth curve, satisfying:
$$\dot{\gamma}'_1(s)= \dot{\Phi}^{-1}(s)\dot{\gamma}_1(\Phi^{-1}(s))=\frac{\dot{\gamma}_1(\Phi^{-1}(s))}{\dot{\gamma}_1(\Phi^{-1}(s))}=1.$$
Thus $\gamma'_1(s)=s$ $\forall s\in [0,c]$ and let $\sigma=\gamma'_2$ so that $\gamma(s)=(s,\sigma(s))$, $(0,\sigma(0))=\gamma'(0)=\gamma(a)=(0,x)$ and $(c,\sigma(c))=\gamma'(c)=\gamma(b)=(t,y)$. Thus, $c=t$, $\sigma(0)=x$, $\sigma(t)=y$ and as $\gamma'$ is still causal then $|\dot{\sigma}(s)|\leq 1$ for every $s\in [0,t]$.
\end{proof}
We recall a property of Riemannian manifolds which will be very useful to us. Clearly it is false for Lorentzian manifolds. 

\begin{prop}[Mean Value Theorem]\label{meanvaluetheorem}
Let $(\Sigma,h)$ be a Riemannian manifold. Then, for any piecewise smooth curve $\sigma\colon[a,b]\rightarrow \Sigma$, (where $a,b\in \mathbb{R}$, $a<b$): $$d(\sigma(a),\sigma(b))\leq L(\sigma)=\int_a^b|\dot{\sigma}(s)|ds\leq (b-a)\sup_{s\in[a,b]}\{|\dot{\sigma}(s)|\}$$
\end{prop}
Note that this implies that, for any such curve, $d(\sigma(t),\sigma(t'))\leq |t'-t|\sup_{s\in[a,b]}\{|\dot{\sigma}(s)|\}$ for any $t,t'\in[a,b]$. Also, as the speed of $\sigma$ is bounded over $[a,b]$ (a compact set), then $\sigma$ is uniformly continuous on $[a,b]$. Similarly, if $\sigma\colon(a,b)\rightarrow \Sigma$ is a smooth curve such that $|\dot{\sigma}|$ is bounded on $(a,b)$, then $\sigma$ is uniformly continuous on $(a,b)$. 

In the following we shall be making use of the concepts of continuously extendible curves and extendible geodesics. These terms are defined in O'Neill~\cite{g}. We also refer the reader to Definition~3.6 and the discussion following in Bullock~\cite{me}. 

We shall shortly need the following theorem from O'Neill~\cite{g} (Lemma 5.8) on the extendibility of geodesics. We quote it here for the reader's convenience:
\begin{thm}\label{extendibilityofgeodesics}
Given $b<\infty$ then a geodesic $\gamma\colon[a,b)\rightarrow M$ in a Lorentzian or Riemannian manifold $M$ is geodesically extendible iff it is (continuously) extendible.
\end{thm}
\begin{lem}\label{extendiblegeodesic}
Let $\gamma\colon[a,b)\rightarrow M$ be a geodesic in a Riemannian manifold, then it is geodesically extendible iff there exists a compact set $C\subseteq M$ s.t. $[\gamma]\colon=\gamma([a,b))\subseteq C$.
\end{lem}
\begin{proof}
If $\gamma$ is geodesically extendible then its extension $\gamma'\colon[a,b]\rightarrow M$ is continuous so $C=\gamma'([a,b])$ is compact. Conversely, if $\gamma([a,b))\subseteq C$ then, since $\gamma$ is a geodesic, it has constant speed and so is uniformly continuous. As $C$ is compact it is complete as a metric space. Thus we can extend $\gamma$ continuously to $[a,b]$ by basic functional analysis. So $\gamma$ is continuously extendible and also geodesically extendible by Theorem~\ref{extendibilityofgeodesics}.
\end{proof}

Note that this lemma is also true in the case of Lorentzian manifolds. The proof can be reached by applying Lemma 1.56, Proposition 3.38 and Lemma 5.8 of O'Neill \cite{g}. In this paper we only need the result in its current form.

\begin{prop}\label{J(K)compact}
Consider the spacetime $(M,g)=(\mathbb{R}\times \Sigma, dt^2-h)$. Let $t\geq 0$ and $K\subseteq\Sigma_0$ be compact. Then the following statements are equivalent:
\begin{enumerate}
\item $C(K,t)$ is compact,
\item $\overline{B_t(0)}\subseteq \epsilon_p\text{ for all } p\in K$,
\item $J(K )\cap\Sigma_t$ is compact.
\end{enumerate}
If Statements~1-3 are true, then $C(K,t)=\bigcup_{p\in K}\exp_p[\overline{B_t(0)}]=J(K)\cap\Sigma_t$.
\end{prop}
\begin{rem}
In order to make sense of this proposition, take note of the following definitions: Given a metric space $(X,d)$ and $K\subseteq X$, then for $t\geq 0$ define $C(K,t):=\{p\in X\text{ such that } d(p,K)\leq t\},$ where $d(p,K):=\inf_{q\in K}\{d(p,q)\}$ (see e.g. Appendix~C of Bullock~\cite{me}). We are implicitly using the metric $d$ on $\Sigma_0$ induced by the Riemannian metric $h$ via Theorem~\ref{Riemmetric}. We define $\epsilon_p\subseteq T_p M$ to be the domain of the exponential map $\exp_p$ at $p$, induced by the Riemannian metric $h$. The set $B_t(0)\subseteq T_p M$ is the open ball of radius $t$ centered on $0\in T_p M$ with respect to the norm induced by $h_p$. Note that, since $C(K,t)$ is given in terms of the metric $d$ induced by the Riemannian metric $h$, then all three expressions $C(K,t), \overline{B_t(0)}$ and $J(K)\cap\Sigma_t$ depend on $h$. Indeed, if a different equivalent metric $d'$ was chosen, then this proposition would be, in general, false. 
\end{rem}
\begin{proof}[Proposition~\ref{J(K)compact}]
(\textit{1}$\Rightarrow$\textit{2}) $C(K,t)$ is compact $\Rightarrow$ $C(p,t)$ is compact for all $p\in K$.
Take $p\in K$ and let $X_p\in\overline{B_t(0)}$, so $|X_p|\leq t$. Let $\sigma$ be the maximal geodesic in $\Sigma$ through $p$ s.t.:
$\dot{\sigma}(0)=X_p$, $\sigma\colon[0,b)\rightarrow \Sigma$ and so $|\dot{\sigma}(s)|=|X_p|\leq t \;\forall s\in[0,b)$. If $b\leq 1$, then $L(\sigma|_0^{t'})=\int_0^{t'}|\dot{\sigma}(s)|ds\leq t't\leq bt\leq t$ and so $d(p,\sigma(s))\leq L(\sigma|_0^{t'})\leq t\;\;\forall s\in[0,b)$. And so $\sigma(s)\in C(p,t)\;\forall s\in[0,b)$ but by Theorem~\ref{extendiblegeodesic}, then $\sigma$ can be extended to a geodesic defined on $[0,b+\epsilon)$, contradiction. Thus $b> 1$ and $X_p\in \epsilon_p\;\forall p\in K$.

(\textit{3}$\Rightarrow$\textit{2}) As in the proof of (\textit{1}$\Rightarrow$\textit{2}), choose $p\in K$, $X_p\in\overline{B_t(0)}$ and $\sigma\colon[0,b)\rightarrow \Sigma$ be the maximal geodesic through $p$ with speed $X_p$ at $p$. If $b\leq 1$, for $c\in[0,b)$, let $\sigma'\colon[0,t]\rightarrow \Sigma$, $\sigma'(s)=\sigma(s\frac{c}{t})$ and $|\dot{\sigma}'(s)|=\frac{c}{t}|\dot{\sigma}(s\frac{c}{t})|\leq c<b\leq 1$.
Thus $(t,\sigma(s))\subseteq J(p)\cap \Sigma_t\subseteq J(K)\cap \Sigma_t \text{ (compact by assumption)}\;\;\forall s\in[0,b)$. Again by Theorem~\ref{extendiblegeodesic}, then $\sigma$ can be extended to a geodesic defined on $[0,b+\epsilon)$, contradiction. Thus $b>1$ and $X_p\in \epsilon_p\text{ for all } p\in K$ such that $|X_p|\leq t$. 

(\textit{2}$\Rightarrow$\textit{1}) $\overline{B_t(0)}\subseteq \epsilon_p\;\;\forall p\in K$ implies $C(p,t)=\exp_p[\overline{B_t(0)}]$, which is Corollary 5.6.4 in Petersen~\cite{k}. Note that since the exponential map is certainly continuous then it follows that $C(p,t)$ is compact $\forall p\in K$. It follows from this that $C(K,t)=\bigcup_{p\in K}C(p,t)$ is compact, as is shown in Proposition~C.7 of Bullock~\cite{me}.

(\textit{1}$\Rightarrow$\textit{3}) Given $q_n\in J(K)\cap\Sigma_t\subseteq C(K,t)$, then by compactness there exists a subsequence $q_{n_k}\rightarrow q\in C(K,t)=\bigcup_{p\in K}C(p,t)$. Thus $q\in C(p,t)=\exp_p[\overline{B_t(0)}]$ for some $p\in K$. So there exists a geodesic $\sigma\colon[0,1]\rightarrow \Sigma$ s.t. $\sigma(0)=p, \sigma(1)=q, |\dot{\sigma}(s)|=|\dot{\sigma}(0)|\leq t$. So define $\sigma'\colon[0,t]\rightarrow \Sigma$, $\sigma'(s)=\sigma(s/t)$, $|\dot{\sigma}'(s)|=\frac{1}{t}|\dot{\sigma}(s/t)|\leq 1$, $\sigma'(0)=p$, $\sigma'(1)=q$ and $(t,q)\in J(K)\cap\Sigma_t$ by Proposition~\ref{causalfuture}, so $J(K)\cap\Sigma_t$ is compact. 

(The final statement) When Statements~\textit{1}-\textit{3} are true then, fixing $p\in K$, we have $C(p,t)=\exp_p[\overline{B_t(0)}]$ from Corollary 5.6.4 in Petersen~\cite{k}. But we have $J(p)\cap\Sigma_t\subseteq C(p,t)= \exp_p[\overline{B_t(0)}]$ for all $t\in \mathbb{R}$. Furthermore, the argument in the proof of (1$\Rightarrow$ 3) shows: $\exp_p[\overline{B_t(0)}]\subseteq J(p)\cap\Sigma_t$. So $C(p,t)=\exp_p[\overline{B_t(0)}]=J(p)\cap\Sigma_t\;\forall p\in K$. Thus $C(K,t)=\bigcup_{p\in K}C(p,t)= \bigcup_{p\in K}\exp_p[\overline{B_t(0)}]= \bigcup_{p\in K} J(p)\cap\Sigma_t=J(K)\cap\Sigma_t$.
\end{proof}

In particular, the content of this proposition is true when $K=\{p\}$, where $p$ is any point in $\Sigma_0$. The usefulness of this proposition arises from the fact that $C(K,t)$ is easier to visualise than $J(K)\cap \Sigma_t$ as Section~\ref{sec:causalstructure (ii)} utilises. 

We recall the notion of a globally hyperbolic spacetime: 
\begin{defn} A spacetime $(M,g)$ is \textbf{globally hyperbolic} if:
\begin{enumerate}
\item It obeys the causality condition: there exist no closed causal curves.			   
\item $J^+ (p)\bigcap J^- (q)$ is compact $\forall p,q \in M$.
\end{enumerate}
\end{defn}
Note, it is shown in Bernal and Sanchez \cite{i} that condition \textit{1} may be equivalently replaced by the ``strong causality condition'', which is the more common definition.

\begin{defn} Given a spacetime $(M,g)$, a \textbf{Cauchy surface} (of $(M,g)$) is a subset $S$ of $M$ that is met exactly once by every inextendible smooth timelike curve in $M$.
\end{defn}
 
To explain the name, we note that every such set is an achronal closed topological embedded hypersurface in $M$ (see Lemma 14.29 in O'Neill \cite{g}). It is an important fact that global hyperbolicity is equivalent to the existence of a Cauchy surface in the spacetime as the following theorem states, which we include for completeness and future reference. Before we state it, we first define the concept of an acausal set in a spacetime, since this notion shall be used in the following theorem.

\begin{defn}[Acausal Set]
A subset $S$ of a spacetime $M$ is called \textbf{acausal} if it is met at most once by any causal curve in $M$. 
\end{defn}
An \textbf{achronal} set is defined similarly with ``any causal curve'' replaced by ``any timelike curve''.

\begin{thm}\label{globhypcauchy}
A spacetime $(M,g)$ is globally hyperbolic iff it possesses a Cauchy surface. If so, then it also possesses a smooth spacelike Cauchy surface. Additionally, if $H$ is a smooth spacelike acausal compact $m$-dimensional embedded submanifold with boundary in $M$, then there exists a smooth spacelike Cauchy surface $S$ in $M$ that contains $H$.
\end{thm}

\begin{proof}
If $S\subseteq M$ is Cauchy surface then $M$ is globally hyperbolic by Corollary 14.39 in O'Neill \cite{g}.
That $M$ is globally hyperbolic implies that it possesses a smooth spacelike Cauchy surface is proved in Theorem 1 in Bernal and Sanchez \cite{q}. For the last statement see Theorem 1.1 of Bernal and Sanchez \cite{h}.
\end{proof}

Note that if $M$ is $n$-dimensional, then in the above theorem: $m\in\{0,...,n-1\}$. In order to aid the understanding of this theorem, we shall shortly give an example to illustrate why we cannot remove the condition of compactness. In order to state this example, we first introduce the concept of Cauchy development of a set, a notion that will be frequently used later.  

\begin{defn}[The Past and Future Cauchy Developments]
Given a subset $S$ of a spacetime $M$, then the \textbf{future Cauchy development} $D^+(S)\subseteq M$ is defined as: 
$$D^+(S)=\left\{p\in M\colon \begin{array}{l}
\text{Every past-inextendible future-pointing smooth} \\
\text{causal curve through $p$ intersects $S$.}\end{array}\right\}$$

The \textbf{past Cauchy development} $D^-(S)$ of $S$ is defined similarly with ``past-inextendible'' replaced by ``future-inextendible''. The \textbf{Cauchy development} $D(S)$ of $S$ is then defined: $D(S)=D^+(S)\cup D^-(S)$.
\end{defn}

Using this definition, let $M=D(\{0\}\times (0,1))$ be an open subset of 2-dimensional Minkowski space and fix $0<|t|<1/2$. Then $H=\{t\}\times (0,1)\cap M$ is a 1-dimensional smooth spacelike, acausal embedded submanifold and so also a submanifold with boundary (just with empty boundary!). However, it is non-compact in $M$ and is contained in no Cauchy surface.
 
The following proposition gives a necessary and sufficient condition for a standard static spacetime to be globally hyperbolic. The content of this proposition is already known. However, we give here an alternative proof. See Lemma A.5.14 in B\"{a}r, Ginoux and Pf\"{a}ffle \cite{e} or Theorem 3.67 in Beem, Ehrlich and Easley \cite{r} for other proofs. 

\begin{prop}
Given the Riemannian manifold $(\Sigma,h)$, then the standard static spacetime $(M,g)=(\mathbb{R}\times\Sigma, dt^2-h)$ is globally hyperbolic iff $(\Sigma,h)$ is a complete Riemannian manifold.
\end{prop}
\begin{proof}
Let $(\Sigma,h)$ be complete. We wish to give two proofs that $(M,g)$ is globally hyperbolic, namely that it possesses a smooth spacelike Cauchy surface and that it satisfies the definition of global hyperbolicity (equivalent by Theorem~\ref{globhypcauchy}).

We wish to show that $\Sigma_0=\{0\}\times\Sigma$ is a smooth spacelike Cauchy surface. Let $\gamma\colon I\rightarrow M$ be a smooth inextendible causal curve w.l.o.g.\ given by $\gamma(t)=(t,\sigma(t))$, where $\sigma\colon I\rightarrow \Sigma$ is a smooth inextendible curve in $\Sigma$ with speed bounded by 1. Then if $I\neq \mathbb{R}$ then as $\sigma$ is uniformly continuous and $\Sigma$ is complete then $\sigma$ can be continuously extended to the closure $\overline{I}$ of $I$ in $\mathbb{R}$, where $I\neq I'$, contradicting the inextendibility of $\sigma$. Thus $I=\mathbb{R}$, $\gamma(0)\in \Sigma_0$, and so any inextendible smooth causal curve passes $\Sigma_0$. Note also the parametrisation: $\gamma(t)=(t,\sigma(t))$ also shows that it must pass $\Sigma_0$ once and only once.

We now show that $M$ satisfies the definition of a globally hyperbolic spacetime. That it is causal follows from the previous argument. We must now show that $J^+(x)\cap J^-(y)$ is compact for all $x,y\in M$. Let $x=(t,p)$ and $y=(t',q)$.
By Proposition~\ref{J(K)compact}, since $(\Sigma,h)$ is complete we know that $J^+(x)\cap\Sigma_s=\{s\}\times C(p,s-t)$ for all $s\geq t$. Thus:
\begin{align*}
J^+(x)\cap J^-(y)&=\left[ \bigcup_{s\geq t}\{s\}\times C(p,s-t)\right]\cap\left[\bigcup_{s'\leq t'}\{s'\}\times C(q,t'-s')\right]\\
&=\bigcup_{s\geq t, s'\leq t'}\{s\}\times C(p,s-t)\cap \{s'\}\times C(q,t'-s')\\
&= \bigcup_{t\leq s\leq t'}\{s\}\times [C(p,s-t)\cap C(q,t'-s)]
\end{align*}
Note that $C(p,s-t)\cap C(q,t'-s)$ is compact since complete Riemannian manifolds obey the Heine-Borel property (Theorem 16 of Petersen \cite{k}). Let $z_n=(s_n,r_n)\in J^+(x)\cap J^-(y)$, so $s_n\in [t,t']$ and $r_n\in [C(p,s_n-t)\cap C(q,t'-s_n)]$. By taking successive subsequences we have that $s_{n_k}\rightarrow s\in [t,t']$ and $r_{n_k}\rightarrow r\in C(p,s-t)\cap C(q,t'-s)$. So $z_{n_k}\rightarrow z=(s,r)\in J^+(x)\cap J^-(y)$ and the latter is compact.

Now for the converse:
If $(\Sigma,h)$ is not complete, then $\epsilon_p\neq T_p\Sigma$ for some $p\in \Sigma$. So $\exists X_p\in T_p\Sigma$ s.t. $|X_p|=R$, $B(0,R)\subseteq \epsilon_p$ and $X_p\nin \epsilon_p$. Consequently, there exists a geodesic $\sigma\colon[0,t)\rightarrow \Sigma$ that is (continuously) inextendible by Theorem~\ref{extendibilityofgeodesics} and has unit speed. Let $x=(0,p)$, $y=(2t,p)$ so $(t,\sigma(s))\in J^+(x)\cap J^-(y)$ $\forall s\in [0,t)$.
If $(M,g)$ is globally hyperbolic, then $J^+(x)\cap J^-(y)$ is compact and $\text{ for all } s\in [0,t)$: $\sigma(s)\in \pi_t^{-1}[\Sigma_t\cap J^+(x)\cap J^-(y)]$, where the RHS is compact in $\Sigma$.  So, $\sigma$ is extendible by Lemma~\ref{extendiblegeodesic}, which is a contradiction. 
\end{proof}
Note that since global hyperbolicity is preserved under conformal transformations then we also have the following result:
 
\begin{lem}\label{staticandglobhyp}
Given the Riemannian manifold $(\Sigma,h)$ and the smooth function $V\in C^\infty(\Sigma)$, $V>0$ then the standard static spacetime $(M,g)=(\mathbb{R}\times\Sigma, V^2 dt^2-h)$ is globally hyperbolic iff $(\Sigma,V^{-2}h)$ is a complete Riemannian manifold.
\end{lem}

We shall now analyse the Cauchy development $D(\Sigma_0)$ of the set $\Sigma_0$ in a standard static spacetime. Note that the Cauchy development of a set may be open (in the case of $D(\{0\}\times\mathbb{R})\subseteq \mathbb{M}^{1+1}$) or closed (as in the case of $D(\{0\}\times[a,b])\subseteq \mathbb{M}^{1+1}$). 
The following proposition states that in a standard static spacetime the Cauchy development of $\Sigma_0$ is open. Thus it is a smooth embedded submanifold and the metric $g$ gives it the structure of a spacetime. Note that this spacetime is also static but not in general standard static (e.g.\ let $M=\mathbb{R}\times (0,1)$ be a strip in Minkowski space with the induced Lorentzian metric. Then $D(\Sigma_0)$ is an open diamond.) In fact, this spacetime $D(\Sigma_0)$ is also globally hyperbolic. 

\begin{prop}
Let $(M,g)=(\mathbb{R}\times\Sigma, V^2 dt^2-h)$ be the standard static spacetime in Definition~\ref{standardstatic} then $\Sigma_0$ is an acausal smooth embedded spacelike hypersurface in $M$. Also, $M$ satisfies the causality condition.
\end{prop}
\begin{proof}
We have shown in Proposition~\ref{causalfuture} that if $\gamma\colon I\rightarrow M$ is a causal curve meeting $\Sigma_0$ (where $I$ is an open interval of $\mathbb{R}$) then, after taking a reparametrisation, we may let $\gamma(t)=(t,\sigma(t))$, where $\sigma\colon I\rightarrow \Sigma$ is a smooth curve with speed bounded by 1 and $0\in I$ and clearly $\gamma$ only passes $\Sigma_0$ once. A similar argument also works for the second statement.
\end{proof}

\begin{prop}\label{DSigmaisglobhyp}
Given any acausal topological hypersurface $S$ in a spacetime $(M,g)$, then $D(S)$ is open in $M$ and $(D(S),g)$ is a globally hyperbolic spacetime. In fact $S$ is a Cauchy surface for $(D(S),g)$.
\end{prop}

\begin{proof}
See Propositions 14.38 and 14.43 of O'Neill \cite{g}.  
\end{proof}

Thus a consequence of the previous two propositions is that given a standard static spacetime $(M,g)=(\mathbb{R}\times\Sigma,V^2dt^2-h)$ then $(D(\Sigma_0),g)$ is a globally hyperbolic spacetime. We now give an explicit expression for $D(\Sigma_0)$ and an alternative proof that it is an open set in $M$.
\begin{prop}\label{DSigmaexplicit} Given a standard static spacetime $(M,g)=(\mathbb{R}\times\Sigma, V^2dt^2-h)$ then the following statements are true:
\begin{enumerate}
\item $D^+(\Sigma_0)=\{(t,p)\in M\colon \;C(p,t) \text{ is compact in } \Sigma, t\geq 0\}$.
\item $D^-(\Sigma_0)=\{(-t,p)\in M\colon\;C(p,t) \text{ is compact in } \Sigma, t\geq 0\}=TD^+(\Sigma_0)$.
\item $D(\Sigma_0)=D^+(\Sigma_0)\cup D^-(\Sigma_0)= \{(t,p)\in M\colon\;C(p,|t|) \text{ is compact in } \Sigma \}$.
\item $D(\Sigma_0)$ is open in $M$,
\end{enumerate}
\end{prop}
where $C(p,t)$ is the closed ball centered on $p$ of radius $t$ in the metric on $\Sigma$ induced by the Riemannian metric $V^{-2}h$ (see Theorem~\ref{Riemmetric}), $T:M\rightarrow M$ is the smooth map: $T(t,p)=(-t,p)$ and $\Sigma_0=\{0\}\times \Sigma$.

Before we prove Proposition~\ref{DSigmaexplicit}, we prove the following very useful result:
\begin{prop}
With the definitions of the previous proposition, let $K\subseteq \Sigma$ and $C(K,t)$ be compact in $\Sigma$, where $t\geq 0$, then $\{t\}\times K\subseteq D^+(\Sigma_0)$.
\end{prop}
\begin{proof}
As usual we let w.l.o.g.\ $V=1$ for simplicity. Let $p\in K$ and $\gamma\colon I\rightarrow \mathbb{R}\times\Sigma$ be an inextendible future-pointing smooth causal curve through $(t,p)$, where $I$ is an open interval of $\mathbb{R}$. By Proposition~\ref{causalfuture} w.l.o.g.\ we can set $\gamma(s)=(s,\sigma(s)) \;\forall s \in I$, where $\sigma\colon I\rightarrow \Sigma $ is a smooth curve with $t\in I$, $\sigma(t)=p$ and $|\dot{\sigma}(s)|\leq 1\;\forall s\in I$. Let $I=(a,b)$ where $b\in\mathbb{R}\cup\{\infty\}$. If $a\geq 0$ then for $s\in (a,t)$: 
$$d(p,\sigma(s))\leq L(\sigma|_s^t)=\int _s^t|\dot{\sigma}(s')|ds'\leq t-s\leq t$$
and so $\sigma(s)\in C(K,t)\;\forall s\in (a,t)$. But since $\sigma$ is a smooth curve in $\Sigma$ with speed bounded by $1$ it is uniformly continuous by the Mean Value Theorem (Theorem~\ref{meanvaluetheorem}) and since it is contained in the compact (and thus complete) set $C(p,t)$, then it can be continuously extended, contradiction. Thus $a<0$ and $\gamma$ passes $\Sigma_0$.   
\end{proof}

\begin{cor}\label{moreDSigma}
Again, with the definitions of the previous propositions, if $C(K,t)$ is compact in $\Sigma$ then $\{s\}\times C(K,t-s)\subseteq D^+(\Sigma_0)$ for all $s\in [0,t]$.
\end{cor}
\begin{proof}
So (by Proposition~C.4 of Bullock~\cite{me}), $C(K,t)=C(C(K,t-s),s)$ is compact. By the previous proposition then $\{s\}\times C(K,t-s)\subseteq D^+(\Sigma_0)$ for all $s\in [0,t]$. 
\end{proof}

\begin{proof}[Proposition~\ref{DSigmaexplicit}]
Again, for simplicity and w.l.o.g.\, assume $V=1$. We start by proving Statement \textit{1}: $$D^+(\Sigma_0)=\{(t,p)\in M\colon\;C(p,t) \text{ is compact in } \Sigma, t\geq 0\}$$
That the RHS is contained in the LHS follows from Corollary~\ref{moreDSigma} with $K=\{p\}$ and $s=t$.
For the converse, let $(t,p)\in D^+(\Sigma_0)$. So $t\geq 0$ and any past-inextendible future-pointing inextendible smooth causal curve through $(t,p)$ passes $\Sigma_0$. Thus, using the symmetry of the spacetime, any future-pointing future-inextendible smooth causal curve through $(t,p)$ passes $\Sigma_{2t}$. Take for instance the curve $\gamma\colon I\rightarrow M$, where $0\in I$,  $\gamma(s)=(t+s,\sigma(s))$ and $\sigma$ is an inextendible geodesic with $\sigma(0)=p$ and $\dot{\sigma}(0)=X_p$, with $|X_p|\leq 1$. Let $I=(a,b)$. The curve $\gamma$ is thus causal and inextendible and so passes $\Sigma_{2t}$ and so $b>t$. Alternatively if $|X_p|\leq t$ then (by Lemma 5.8 (Rescaling Lemma) in Lee \cite{p}) $b>1$ and so $\overline{B(0,t)}\subseteq \epsilon_p$ or, by Proposition~\ref{J(K)compact}, $C(p,t)$ is compact in $\Sigma$. 

Statements \textit{2} and \textit{3} follow. Now for the proof of Statement~\textit{4} that $D(\Sigma_0) $ is open. Let $(t,p)\in D(\Sigma_0)$ w.l.o.g.\ $t\geq 0$. So $C(p,t)$ is compact in $\Sigma$ and from Corollary~C.6 of Bullock~\cite{me}, there exists $\epsilon>0$ s.t. $C(p,t+\epsilon)$ is also compact. We propose that:
$$\left(-\left(t+\frac{\epsilon}{2}\right),t+\frac{\epsilon}{2}\right)\times B\left(p,\frac{\epsilon}{2}\right)\subseteq D(\Sigma_0).$$
This follows by showing that if $(s,q)\in (-(t+\frac{\epsilon}{2}),t+\frac{\epsilon}{2})\times B(p,\frac{\epsilon}{2})$ then $C(q,s)$ is compact (the result then follows from the description of $D(\Sigma_0)$ just proven). Firstly, we can set w.l.o.g.\ $s\in[0,t+\frac{\epsilon}{2})$, $d(p,q)<\frac{\epsilon}{2}$. But $r\in C(q,s)\Rightarrow d(r,q)\leq s$ and so:
$$d(r,p)\leq d(r,q)+d(q,p)<s+\frac{\epsilon}{2}<t+\frac{\epsilon}{2}+\frac{\epsilon}{2}=t+\epsilon.$$
So $C(q,s)\subseteq C(p,t+\epsilon)$ and as the RHS is compact then so is $C(q,s).$
\end{proof}

It is well known that given a globally hyperbolic spacetime and smooth initial data of compact support defined on a smooth spacelike Cauchy surface then the Klein-Gordon equation can be solved uniquely with respect to this data:
\begin{thm}[Existence and Uniqueness of Classical Solutions on Globally Hyperbolic Spacetimes with respect to compactly supported initial data]\label{globhypsolution1}(B\"{a}r et al.\ \cite{e} Theorem 3.2.11) Let $(M,g)$ be a globally hyperbolic spacetime with smooth, spacelike Cauchy surface $S$. Then the Klein-Gordon equation has a well-posed initial value formulation, that is, given data
$\phi_0, \dot{\phi}_0 \in C_0^{\infty}(S)$ then there exists a unique solution $\psi\in C^\infty(M)$ to:
\begin{align*}
 &(\largesquare _g + m^2) \psi = 0\\
 &\psi|_S=\phi_0 \\
 &\nabla _{n} \psi|_S = \dot{\phi}_0, 
\end{align*}
where $n$ is the unique unit smooth future-pointing timelike vector field along $S$ normal to $S$. Moreover: $$\supp\psi\subseteq J(K)$$
where $K=\supp\phi_0\bigcup \supp\dot{\phi}_0$.
\end{thm}
Note that there exists along any smooth spacelike surface $S$ in a spacetime $M$ such a smooth vector field $n$ along $S$ normal to $S$ (the smooth vector field $n$ is not to be confused with the dimension of the spacetime). For completeness, this is proven in Proposition~E.1 of Bullock~\cite{me}. Note that the orientability of $M$ or $S$ is not assumed. 
 
We shall use a modification of this theorem in the next section, that is, we can drop the condition on the data of being of compact support. This is shown e.g.\ in Theorem~B.1 of Bullock~\cite{me}.

\section{The Existence of Wald solutions}\label{existence1}
In this section we show how to construct our solution to the Klein-Gordon equation from the vector-valued function $t\rightarrow [\phi_t]$. This section is strongly based on the paper by Wald~\cite{b}, but is extended in the following aspects. The more recent result by Bernal and Sanchez~\cite{h} on the extendibility of subsets of the spacetime to smooth spacelike Cauchy surfaces in globally hyperbolic spacetimes (the second half of Theorem~\ref{globhypcauchy}) is needed to complete the proof on the existence of Wald solutions. We also extend Wald's proof to the case of acceptable s.a.e.s. The reference for the results on globally hyperbolic spacetimes is, as usual, B\"{a}r et al.~\cite{e}. This section is of great importance to us as it proves that the construction of Section~\ref{candidatesolutions} defines a smooth solution to the Klein-Gordon equation. We answer in the next section the question of its uniqueness.

We start with a theorem concerning the agreement between our solution (Equation~\eqref{stuff2}, p.\pageref{stuff2}) to the Hilbert space version of the Klein-Gordon equation (Equation~\ref{stuff}, p.\pageref{stuff}) and that arising from an application of Theorem~\ref{globhypsolution1} (or rather the generalisation mentioned thereafter):
\begin{thm} \label{waldglobsolutionagreement} Given initial data $\phi_0,\dot{\phi}_0\in \chi_E$, 
where $A_E$ is an acceptable s.a.e.\ of $A$, choose $\phi_t\in \chi_E$ s.t. $[\phi_t]=C(t,A_E)[\phi_0]+S(t,A_E)[\dot{\phi}_0]$. If we define the function $\phi$ on $M$ by: $\phi(t,x)=\phi_t(x)$, 
and let $\psi$ be the unique smooth solution in $D(\Sigma_0)$ satisfying this smooth Cauchy data according to Theorem~B.1 of Bullock~\cite{me} then $\phi=\psi$ in $D(\Sigma_0)$ and, in particular, $\phi|_{D(\Sigma_0)}$ is smooth and solves the Klein-Gordon equation there.
\end{thm}

Note that if $A_E$ is bounded-below then $A_E^-$ is bounded and $\chi_E=\{f\in C^\infty(\Sigma)\text{ s.t. } [f]\in D(A_E^\infty)\}$.

This is proven by contradiction. The proof is due to Wald~\cite{b} but completed (by reference to a more recent result of Bernal and Sanchez~\cite{h} on the existence of smooth spacelike Cauchy surfaces) and extended to the case of acceptable s.a.e.s dealt with in this paper. The proof is included for completeness.

\begin{prop}
If there exists $t_1$ such that $\phi\neq\psi$ everywhere in a non-null set in $\Sigma_{t_1}\cap D(\Sigma_0)$, then there exists a compact set $H$ in $\Sigma_{t_1}\cap D(\Sigma_0)$ and a smooth spacelike Cauchy surface $S$ for $D(\Sigma_0)$ s.t. $H\subseteq S$ and $\vol_h\{(t_1,x)\in H\colon \;\psi(t_1,x)\neq\phi(t_1,x)\}>0$.
\end{prop}

\begin{SCfigure} \centering
\begin{tikzpicture}[scale=0.75] 
\draw[dash pattern=on 8pt off 8pt] (0,0) -- (0,5) node[label=right:$\mathbb{R}\times\Sigma$]{};
\draw[dash pattern=on 8pt off 8pt] (5,0) -- (5,5);
\draw(5,0) node[label=right:$\Sigma_0$]{};
\draw(5,2) node[label=right:$\Sigma_{t_1}$]{};
\draw(0,2)--(5,2);
\draw(0,0)--(5,0);
\draw(2.5,1) node[label=below:$D(\Sigma_0)$]{};
\draw[dash pattern=on 8pt off 8pt](0,0) -- (2.5,4.5) -- (5,0);
\draw (2.45,1.9)--(2.5,1.9)--(2.5,2.1)--(2.45,2.1);
\draw (2.05,1.9)--(2,1.9)--(2,2.1)--(2.05,2.1);
\draw (0,0) .. controls (0.5,0.9) and (1,2)..(2,2);
\draw (2.5,2) .. controls (4,2) and (4.5,0.9)..(5,0);
\draw(2.25,2) node[label=above:H]{};
\draw(4,1.5) node[label=left:S]{};
\end{tikzpicture}
\caption{$\phi=\psi$ in $D(\Sigma_0)$, where $(M,g)=(\mathbb{R}\times (0,1),dt^2-dx^2)$\vspace{6em}}
\end{SCfigure}
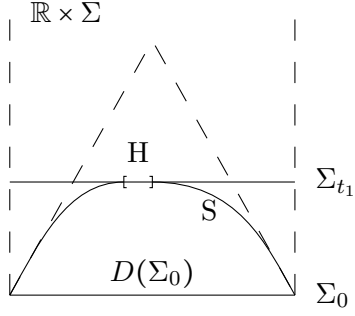

\begin{proof} So, by assumption there exists $t_1\in\mathbb{R}$ such that 
\begin{equation}\label{eq}
\vol_h\{(t_1,x)\in \Sigma_{t_1}\cap D(\Sigma_0)\colon \psi(t_1,x)\neq\phi(t_1,x)\}>0. 
\end{equation}
Now we construct a smooth compact embedded submanifold with boundary $H$ of $\Sigma_{t_1}\cap D(\Sigma_0)$ s.t.  $$\vol_h\{(t_1,x)\in H\colon \;\psi(t_1,x)\neq \phi(t_1,x)\}>0.$$
Firstly, let $U=\{(t_1,x)\in \Sigma_{t_1}\cap D(\Sigma_0)\colon \;\psi(t_1,x)\neq\phi(t_1,x)\}$, so $\vol_h(U)>0$.   
Since any manifold has a countable atlas (see e.g.\ Warner~\cite{n}, Lemma 1.9), then there exists such an atlas $(V_n,\phi_n)_{n\geq 0}$ of $\Sigma_{t_1}\cap D(\Sigma_0)$ with $U=\bigcup_{n\geq 0}U\cap V_n$ and $\vol_h(U)\leq \sum_{n\geq 0}\vol_h(U\cap V_n)$ and so there must be one chart $(V_n,\phi_n)$ s.t. $\vol_h(U\cap V_n)>0$. Let $(V,\phi)=(V_n,\phi_n)$.

Secondly, by a similar argument, as $\phi(V)$ is a open subset of $\mathbb{R}^N$ and any open subset of $\mathbb{R}^N$ can be covered by a countable number of open balls then there exists an open ball $B=\{x\in \mathbb{R}^N \text{ s.t. } ||x||<r\}$ (w.l.o.g.\ centered at 0) s.t. $B\subseteq \phi(V)$ and $\vol_h(U\cap \phi^{-1}(B))>0$. 

Lastly, since $B=\{x\in \mathbb{R}^N \text{ s.t. } ||x||<r\}$ is covered by the countable collection of closed balls $C_n=\{x\in \mathbb{R}^N \text{ s.t. } ||x||\leq r_n\}$ where $(r_n)_{n\geq 1}$ is any sequence of positive reals s.t. $r_n\nearrow r$ and as before there must exist $n\geq 1$ s.t. $\vol_h(U\cap \phi^{-1}(C_n))>0$. Let $H=\phi^{-1}(C_n)$ be the desired smooth compact submanifold with boundary of $\Sigma_{t_1}\cap D(\Sigma_0)$. Since $H\subseteq\Sigma_{t_1}\cap D(\Sigma_0)$ and the latter is a smooth spacelike acausal embedded submanifold then $H$ is a smooth compact acausal spacelike embedded submanifold with boundary of the spacetime $(D(\Sigma_0),g)$. The reason for this construction is that it allows us to apply Theorem~\ref{globhypcauchy}. Thus there exists a smooth spacelike Cauchy surface $S$ of $(D(\Sigma_0),g)$ which contains $H$.
\end{proof}

Now let $\dot{f}_{t_1}$ be a smooth compactly supported function on $S$ with support in $S\cap \Sigma_{t_1}$ such that $\dot{f}_{t_1}\geq 0$ and $\dot{f}_{t_1}=1$ on $H$. Thus:
$$\int_{S\cap\Sigma_{t_1}} \dot{f}_{t_1}(\psi-\phi)V^{-1}d\vol_h\neq 0$$ 
and define $f$ to be the unique smooth solution to the Klein-Gordon equation on $D(\Sigma_0)$ with Cauchy data $(0,\dot{f}_1)$ on $S$, according to Theorem~\ref{globhypsolution1}. 

We define $F\colon[0,t_1]\times \Sigma\rightarrow \mathbb{R}$ as:
$$F(p)=\left\{\begin{array}{ll}
f(p), &p\in [0,t_1]\times\Sigma\cap D(\Sigma_0).\\
0, &\text{otherwise.}\end{array}\right.$$
\begin{prop} $F$ satisfies the following:
\begin{enumerate}
\item $\supp F$ is compact in $[0,t_1]\times\Sigma\cap D(\Sigma_0)$.
\item It is compactly supported on each $\Sigma_t\cap D(\Sigma_0)$ for $0\leq t\leq t_1$.
\item $F \in C^\infty([0,t_1]\times \Sigma)$ (as a smooth manifold with boundary).
\item $(\largesquare_g +m^2)F=0$ (as an element of $C^\infty([0,t_1]\times \Sigma)$). 
\item $\supp (\partial_t F)\cap\Sigma_{t_1}= \supp \dot{f}_{t_1}\subseteq S\cap \Sigma_{t_1}$ and $\partial_t F(p)=\dot{f}_{t_1}(p) \text{ for } p\in S\cap \Sigma_{t_1}$. 
\item $F|_{\Sigma_{t_1}}=0$.
\end{enumerate}
\end{prop}

\begin{proof} 
By construction $\supp\dot{f}_{t_1}$ compact in $S$ and contained in $D(\Sigma_0)\cap\Sigma_{t_1}$. As all hypersurfaces concerned are embedded, then all have their topologies induced from that of $M$ and thus $\supp\dot{f}_{t_1}$ compact in $D(\Sigma_0)$. 

But, the causal past of a compact set intersected with the causal future of a Cauchy surface $S$ (in a globally hyperbolic spacetime) is always compact (see Corollary A.5.4 of B\"{a}r et al.\ \cite{e}). Thus: $$J^-_{D(\Sigma_0)}(\supp \dot{f}_1)\cap J^+_{D(\Sigma_0)}(\Sigma_0)\text{ is compact in }D(\Sigma_0).$$
So, $J^-_{D(\Sigma_0)}(\supp \dot{f}_1)\cap[0,t_1]\times \Sigma \text{ is compact in }D(\Sigma_0)$ and so also in $[0,t_1]\times\Sigma\cap D(\Sigma_0)$. Thus, $\supp F\subseteq \supp f \cap[0,t_1]\times \Sigma$ is compact in $[0,t_1]\times\Sigma\cap D(\Sigma_0)$ and Statement \textit{1} is proved. Statements \textit{2} and \textit{3} follow directly from \textit{1}. Now, since, by definition, $F$ is locally equal to either $f$ or $0$, where both are smooth solutions to the Klein-Gordon equation, then Statement~\textit{4} follows. Statements~\textit{5} and \textit{6} result straight from the definitions of $F$ and $f$.
\end{proof}

\begin{thm}\label{waldglobhypagreement}
The functions $\phi$ and $\psi$ are equal on $D(\Sigma_0)$.
\end{thm}

\begin{proof}
If there exists $t_1$ such that $\phi\neq\psi$ everywhere in a non-null set in $\Sigma_{t_1}\cap D(\Sigma_0)$, construct $H, S$ and $F$ as above. Now define: 
\begin{equation}
c(t)=\int_{\Sigma_t} V^{-1}\left[ F\left(\frac{\partial \psi}{\partial t}-\frac{d\phi_t}{dt}\right)-\frac{\partial F}{\partial t}(\psi-\phi_t)\right]d\vol_h
\end{equation}

Clearly since $\phi_t$ and $\frac{d\phi_t}{dt}$ are only defined a.e.\ in $\Sigma$ we should point out that any other choices in the same respective equivalence classes would yield an identical value of $c$. As both functions are in $\mathcal{L}^2(\Sigma)$, then multiplying by the smooth functions of compact support, $F$ and $\frac{\partial F}{\partial t}$, we obtain an element of $\mathcal{L}^1(\Sigma)$.

The smooth function $\psi$ is only defined in $D(\Sigma_0)$ and so on each hypersurface $\Sigma_t\cap D(\Sigma_0)$, $\psi$ and $\frac{\partial \psi}{\partial t}$ are smooth functions but as $F$ and $\frac{\partial F}{\partial t}$ are compactly supported smooth functions on each $\Sigma_t\cap D(\Sigma_0)$, then $f\frac{\partial \psi}{\partial t}$ and $\frac{\partial F}{\partial t}\psi$ are easily definable and smooth on each $\Sigma_t$, $t\in[0,t_1]$. Indeed they are of compact support also so they are integrable on $\Sigma_t$ (since we are dealing with a Radon measure). Thus: 

\begin{align*}
\frac{dc}{dt}&= \int_{\Sigma_t} V^{-1}\left[ \frac{\partial F}{\partial t}\left(\frac{\partial \psi}{\partial t}- \frac{d\phi_t}{dt}\right)+F\left(\frac{\partial^2 \psi}{\partial t^2}- \frac{d^2\phi_t}{dt^2}\right) -\frac{\partial^2 F}{\partial t^2}(\psi-\phi_t)- \frac{\partial F}{\partial t}\left(\frac{\partial \psi}{\partial t}-\frac{d\phi_t}{dt}\right)\right] d\vol_h\\
&=\int_{\Sigma_t} V^{-1}\left[ F\frac{\partial^2 \psi}{\partial t^2}-\frac{\partial^2 F}{\partial t^2}\psi\right] d\vol_h  -\int_{\Sigma_t} V^{-1}\left[ F\frac{d^2\phi_t}{dt^2} - \frac{\partial^2 F}{\partial t^2}\phi_t\right] d\vol_h  \\
&=\int_{\Sigma_t} V^{-1}\left[ FVD^i(VD_i\psi)-VD^i(VD_i F)\psi\right] d\vol_h  -\int_{\Sigma_t} V^{-1}\left[ F\frac{d^2\phi_t}{dt^2} - \frac{\partial^2 F}{\partial t^2}\phi_t\right] d\vol_h \\ 
&=\int_{\Sigma_t} \left[ FD^i(VD_i\psi)-D^i(VD_i f)\psi\right] d\vol_h  -\int_{\Sigma_t} V^{-1}\left[ F\frac{d^2\phi_t}{dt^2} - \frac{\partial^2 F}{\partial t^2}\phi_t\right] d\vol_h \\ 
&=\int_{\Sigma_t} \left[ -(D^iF)(VD_i\psi)+(VD_i F)(D^i\psi)\right] d\vol_h  -\int_{\Sigma_t} V^{-1}\left[ F\frac{d^2\phi_t}{dt^2} - \frac{\partial^2 F}{\partial t^2}\phi_t\right] d\vol_h \\ 
&=\int_{\Sigma_t} V^{-1}\left[ -F\frac{d^2\phi_t}{dt^2} + \frac{\partial^2 F}{\partial t^2}\phi_t\right] d\vol_h \\ 
&=\langle F,A_E \phi_t \rangle -\langle A_E F,\phi_t\rangle\\
&=0.
\end{align*}
\normalsize
But, $\psi|_{\Sigma_0}=\phi_0$ and $\frac{\partial \psi}{\partial t}|_{\Sigma_0}=\dot{\phi}_0$, so $c(0)=0$\\
and since $F|_{t_1}=0$ by definition, we have:
\begin{align*}
c(t_1)&= -\int_{\Sigma_t} V^{-1}\dot{F}_{t_1}(\psi-\phi_t)d\vol_h\\
&\neq 0.
\end{align*}
However $c\in C^1[0,t_1]$ and so this last statement contradicts the Intermediate Value Theorem, yielding that $\phi=\psi$ a.e. in $D(\Sigma_0)\cap\Sigma_t$ for all $t$. Since $\phi$ and $\psi$ are continuous, then $\phi=\psi$ in $D(\Sigma_0)\cap\Sigma_t$ for all $t$ and so $\phi=\psi$ in $D(\Sigma_0)$.
\end{proof}

Thus we have proven Theorem~\ref{waldglobsolutionagreement}. We shall now show that $\phi$ solves the Klein-Gordon equation everywhere in $M$.
\begin{thm}[Existence of Wald Solutions]\label{waldsolutions}
Let $A_E$ be an acceptable s.a.e.\ of $A$. Given any pair of functions $\phi_0,\dot{\phi}_0\in \chi_E$, for each $t\in\mathbb{R}$ define $\phi_t\in \chi_E$ uniquely by: $[\phi_t]=C(t,A_E)[\phi_0]+S(t,A_E)[\dot{\phi}_0]$ and define the function $\phi$ on $M$ as $\phi(t,x)=\phi_t(x)$, where $\phi_t\in C^\infty(\Sigma)$. This function is smooth, solves the Klein-Gordon equation and satisfies the Cauchy data $(\phi_0,\dot{\phi}_0)$, that is $\phi|_{\Sigma_0}=\phi_0$, $\partial_t \phi|_{\Sigma_0}= \dot{\phi}_0$. 
\end{thm}

\begin{proof}
Given $p=(t_1,x)\in M$, we wish to find an open neighbourhood of $p$ in $M$ in which $\phi$ is smooth and satisfies the Klein-Gordon equation. We begin by reformulating our vector-valued solution.  We propose that:
\begin{align*}
[\phi_t]&= C(t,A_E)[\phi_0]+S(t,A_E)[\dot{\phi}_0]\\
&=C(t-t_1+t_1,A_E)[\phi_0]+S(t-t_1+t_1,A_E)[\dot{\phi}_0]\\
&=[C(t-t_1,A_E)C(t_1,A_E)-A_E S(t-t_1,A_E)S(t_1,A_E)][\phi_0]\\
&\;\;\;\;\;\;+[S(t-t_1,A_E)C(t_1,A_E)+C(t-t_1,A_E)S(t_1,A_E)][\dot{\phi}_0]\\
&=C(t-t_1)[C(t_1,A_E)[\phi_1]+S(t_1,A_E)[\dot{\phi}_0]]\\
&\;\;\;\;\;\;+S(t-t_1,A_E)[-A_E S(t_1,A_E)[\phi_0]+C(t_1,A_E)[\dot{\phi}_0]]\\
&= C(t-t_1,A_E)[\phi_{t_1}]+S(t-t_1,A_E)[\dot{\phi}_{t_1}].
\end{align*}
Here, we have used the identities: 
\begin{align*}
C(t_1+t_2,A_E)&=C(t_1,A_E)C(t_2,A_E)-A_E S(t_1,A_E)S(t_2,A_E) \\
S(t_1+t_2,A_E)&=S(t_1,A_E)C(t_2,A_E)+ C(t_1,A_E)S(t_2,A_E) 
\end{align*}
on $D(A_E)$. But $\phi_{t_1},\dot{\phi}_{t_1}\in \chi_E$ and Theorem~\ref{waldglobsolutionagreement} can be applied to this data to show that $\phi$ is smooth in the open neighbourhood $D(\Sigma_{t_1})$ of $p$ and satisfies the Klein-Gordon equation there.  
\end{proof}

\section{Uniqueness of Wald Solutions}\label{sec:uniquenessofwaldsolutions}

We so far have concerned ourselves with constructing a class of solutions to the Klein-Gordon equation on standard static spacetimes. Our set of prescriptions is parametrised by acceptable s.a.e.s $A_E$ of the linear operator $A$ on the (real or complex) Hilbert space $L^2(\Sigma,V^{-1}d\vol_h)$. For each such linear operator $A_E$ we show that the solution to the Klein-Gordon equation w.r.t.\ chosen Cauchy data it generates is unique up to some conditions yet to be stated. We will use this result to define a vector space of solutions, corresponding to each acceptable s.a.e.\ $A_E$.

\begin{thm}[Uniqueness of Solutions (i)]\label{uniqueness1}
Let $A$ be the symmetric linear operator on the (real or complex) Hilbert space $L^2(\Sigma, V^{-1}d\vol_h)$, defined by: $D(A)=[C_0^\infty(\Sigma)]$, $A([\phi])=[(-VD^iVD_i+m^2V^2)\phi]$ for $\phi\in C_0^\infty(\Sigma)$. 
Let $A_E$ be an acceptable s.a.e.\ of $A$ and if $\Psi\in C^2(M)$ satisfies $(\largesquare_g+m^2) \Psi=0$, $\Psi|_{\Sigma_0}=\partial_t\Psi|_{\Sigma_0}=0$, $[\pi_t^*(\Psi|_t)]\in D(A_E)$ and $[\pi_t^*(\partial_t\Psi|_t)]\in L^2(\Sigma, V^{-1}d\vol_h)$ for all $t$ 
(where $\pi_t^*$ is the pull-back of the map $\pi_t\colon\Sigma\rightarrow \Sigma_t$), then $\Psi=0.$
\end{thm}


We start with a proposition, which has its roots in distribution theory on arbitrary Riemannian manifolds.

\begin{prop} Take $A$ and $A_E$ as above. If $\phi\in C^2(\Sigma)$ such that $[\phi]\in D(A_E)$, then $A_E [\phi]=[(-VD^iVD_i+m^2V^2)\phi]$.
\end{prop}

\begin{proof}
We know (already stated on p.\pageref{adjointofA}), that the adjoint $A^*$ of the linear operator $A$ is given by:
$D(A^*)=\{\phi\in L^2(\Sigma, V^{-1}d\vol_h)\text{ s.t. } A\phi\in L^2(\Sigma, V^{-1}d\vol_h)\},$ since $A$ is formally self-adjoint with respect to the smooth measure $V^{-1}d\vol_h$, which is proven in Proposition~D.10 of Bullock~\cite{me}. We can strengthen that proposition to the following case:
$$\int_\Sigma (A\phi)\theta V^{-1} d\vol_h=\int_\Sigma \phi(A\theta) V^{-1} d\vol_h,$$for all $\phi\in C^2(\Sigma)$ and $\theta\in C_0^\infty(\Sigma)$, since $A$ is of second order and commutes with complex conjugation. The proof is similar. Then, if $\phi\in C^2(\Sigma)$ and $[\phi]\in D(A^*)$, we have: $$A^*[\phi](\theta)=\int_\Sigma\phi (A\theta) V^{-1}d\vol_h = \int_\Sigma (A\phi) \theta V^{-1}d\vol_h=[A\phi](\theta),$$ where $A^*[\phi]$ is meant distributionally. Therefore $A^*[\phi]=[A\phi]$. Lastly, since $A_E$ is a s.a.e. of $A$, then $A\leq A_E$ and we have: $A_E\leq A^*$. So, $A_E$ is the restriction of $A^*$ to space $D(A_E)$. Therefore, if $\phi\in C^2(\Sigma)$ and $[\phi]\in D(A_E)$, then $A_E[\phi]=A^*[\phi]=[A\phi]$.
\end{proof}

\begin{proof}[Theorem~\ref{uniqueness1}] We use a proof by contradiction.
Firstly, we point out that if $(\largesquare_g+m^2)\Psi=0$, then $\partial^2_t\Psi=-A\Psi$. But as $\pi_t^*(\Psi|_t)\in D(A_E)$, by the previous proposition: $A(\pi_t^*(\Psi|_t))=A_E (\pi_t^*(\Psi|_t))\in L^2(\Sigma, V^{-1}d\vol_h)$
and thus $\pi_t^*(\partial^2_t\Psi|_t)\in L^2(\Sigma,V^{-1}d\vol_h)$ also. If $\Psi\neq 0$, then there exists $t_1\in\mathbb{R}$ such that $\Psi|_{\Sigma_{t_1}}\neq 0$. Let $\dot{f}_{t_1}\in C_0^\infty(\Sigma_{t_1})$ such that $\int_{\Sigma_{t_1}}\dot{f}_{t_1}\Psi V^{-1}d\vol_h\neq 0$ and let $f_t=S(t-t_1,A_E)(\pi^*_{t_1}\dot{f}_{t_1})$ be the vector-valued function. According to Theorem~$\ref{waldsolutions}$ on the existence of smooth Wald solutions, this function can be represented by the smooth solution $f\in C^\infty(M)$ to the Cauchy problem with smooth initial data $(0,\dot{f}_{t_1})$, of compact support on $\Sigma_{t_1}$. We now evaluate the symplectic form at our two solutions $\Psi$ and $f$:
\begin{align*}
c(t)&=\int_{\Sigma_t} [\partial_t f \Psi-f \partial_t\Psi ]V^{-1}d\vol_h\\
&=\int_\Sigma\left[\pi^*_t(\partial_t f|_t)\pi^*_t(\Psi_t)-\pi^*_t(f|_t)\pi^*_t(\partial_t\Psi_t)\right]V^{-1}d\vol_h.
\end{align*}
Then clearly $c(t_1)\neq 0$ and $c(0)=0$ but:
$$\frac{dc(t)}{dt}=\int_{\Sigma_t}[\partial^2_t f \Psi-f\partial^2_t \Psi]=-\langle A f, \Psi\rangle+\langle f,A\Psi\rangle=-\langle A_E f,\Psi\rangle+\langle f,A_E \Psi\rangle=0,$$
which is a contradiction.   
\end{proof}

\begin{lem}[Uniqueness of Solutions (ii)]
Let $A_E$ be an acceptable s.a.e.\ of A. Given two solutions $\Psi_1,\Psi_2\in C^2(M)$ of the Klein-Gordon equation $(\largesquare_g +m^2)\Psi_i=0$, 
such that $\Psi_1|_{\Sigma_0}=\Psi_2|_{\Sigma_0}$, $\partial_t\Psi_1|_{\Sigma_0}=\partial_t\Psi_2|_{\Sigma_0}$ and for all $t\in\mathbb{R}$ and $i=1,2$: $[\pi_t^*(\Psi_i|_t)]\in D(A_E) \text{ and }[\pi_t^*(\partial_t\Psi_i|_t)]\in L^2(\Sigma,V^{-1}d\vol_h),$
then $\Psi_1=\Psi_2$.
\end{lem}

\begin{proof}
Let $\Psi=\Psi_1-\Psi_2$, then $\Psi\in C^2(M)$ and satisfies the conditions of the previous proposition since all operations concerned are linear and $D(A_E)$ and $L^2(\Sigma, V^{-1}d\vol_h)$ are vector spaces. Thus $\Psi=0$.
\end{proof}

We note here the following trivial generalisation, the proof of which is similar to those previous. It will be this result that will be of use in Section~\ref{sec:supportofwaldsolutions} in describing the support of the Wald solution $\phi$.

\begin{lem}[Uniqueness of Solutions (iii)]\label{uniquenessiii}
Let $A_E$ be an acceptable s.a.e.\ of A. Given two solutions $\Psi_1,\Psi_2\in C^2([t_1,t_2)\times\Sigma)$ of the Klein-Gordon equation $(\largesquare_g +m^2)\Psi_i=0$ 
such that $\Psi_1|_{\Sigma_{t_1}}=\Psi_2|_{\Sigma_{t_1}}$, $\partial_t\Psi_1|_{\Sigma_{t_1}}=\partial_t\Psi_2|_{\Sigma_{t_1}}$ and for all $t\in[t_1,t_2)$ and $i=1,2$ we have: $[\pi_t^*(\Psi_i|_t)]\in D(A_E) \text{ and }[\pi_t^*(\partial_t\Psi_i|_t)]\in L^2(\Sigma,V^{-1}d\vol_h),$
then $\Psi_1=\Psi_2$.

\end{lem}

Using Theorems~\ref{waldsolutions} and~\ref{uniqueness1} on the existence and uniqueness of solutions to the Klein-Gordon equation, we will find it useful to define a vector space of solutions, for each acceptable s.a.e.\ $A_E$ of $A$. We show that it can be given a natural symplectic structure in Section~\ref{sec:sympform}. It's this structure that is required for the construction of the Weyl-algebra, however we will not be concerned with quantisation in this paper. 

\begin{defn}[Space of Solutions]\label{spaceofsolutions}
Given an acceptable s.a.e.\ $A_E$ of $A$, define the \textbf{space of solutions}, $S_E$ to be:
$$S_E=\{\phi\in C^\infty(M)\colon \; (\largesquare_g +m^2)\phi =0, \pi_t^{-1}(\phi_t),\pi_t^{-1}(\dot{\phi}_t)\in \chi_E \text{ for all } t\}$$
\end{defn}
\begin{prop}\label{spaceofsolutionsiso}
We have the linear isomorphism: $\Psi\colon \chi_E\times \chi_E \rightarrow S_E$, defined by $\Psi(\phi_0,\dot{\phi}_0)=\phi$, where $\phi$ is constructed using Theorem~\ref{waldsolutions} on the existence of Wald solutions. 
\end{prop}

\begin{proof}
Clearly $\Psi$ is linear. Surjectivity follows since, if $\psi\in S_E$, then $\psi_0,\dot{\psi}_0\in \chi_E$. Let $\phi$ be the Wald solution, satisfying the Cauchy data $(\psi_0,\dot{\psi}_0)$. Then $\psi$ and $\phi$ satisfy all the conditions of Theorem~\ref{uniqueness1} on uniqueness and so $\psi=\phi$.   
\end{proof}

\section{Causal Structure of Standard Static Spacetimes (ii)}\label{sec:causalstructure (ii)}

We shall in Section~\ref{sec:supportofwaldsolutions} further analyse some of the properties of our constructed solutions to the Klein-Gordon equation. However, we must first prove some basic properties of the causal structure of standard static spacetimes. One apparently simple result of this section is that if $K$ is a compact subset of $\Sigma_0$ then for all sufficiently small $t$, $J^+(K)\cap \Sigma_t$ is compact in $\Sigma_t$. It will be this result and the adapted uniqueness result of Lemma~\ref{uniquenessiii} which will prove useful in the next section. We shall also need to prove more properties of $J^+(K)$ to be used in Section~\ref{sec:supportofwaldsolutions}.

For all the results of this section, let $(M,g)=(\mathbb{R}\times\Sigma, V^2 dt^2-h)$ be a standard static spacetime as in Definition~\ref{standardstatic}. However, in all the statements we can set w.l.o.g. $V=1$, since both the Cauchy development and causal future of a set in a spacetime are identical for conformally related metrics.

\begin{prop}\label{tinfinity}
Let $K\subseteq \Sigma_0$ be a compact set. If $J(K)\cap \Sigma_t$ is compact then $J(K)\cap \Sigma_{t'}$ is compact for all $|t'|\leq|t|$. Define: 
$$t^\infty(K)\colon= \sup\{t\geq 0\colon J^+(K)\cap \Sigma_t \text{ is compact in }\Sigma_t\}.$$ 
Then $t^\infty(K)\in(0,\infty]$. Furthermore, the following are true:
\begin{enumerate}
\item $J(K)\cap \Sigma_t$ is compact for all $|t|<t^{\infty}(K)$. 
\item If $t^\infty(K)<\infty$ then  $J(K)\cap \Sigma_t$ is not compact for all $|t|\geq t^{\infty}(K)$.
\item If $\Sigma$ is complete, then $C(K,t)$ is compact for all $t$ and $t^\infty(K)=\infty$. 
\item If $t^\infty(K)=\infty$ for any non-empty compact set $K$, then $\Sigma$ is complete.
\end{enumerate}
\end{prop}
Note that $\Sigma$ is complete as a metric space iff geodesically complete by the Hopf-Rinow Theorem (see e.g.\ Theorem 6.13 Lee \cite{p}). If so, then $\Sigma$ obeys the Heine-Borel property, that is $K\subseteq \Sigma$ is compact iff $K$ is closed and bounded (see e.g.\ Theorem 16 in Petersen \cite{k}).
\begin{proof}
Let $t\geq 0$. If $J(K)\cap \Sigma_t$ is compact, then, by Proposition~\ref{J(K)compact}, $J(K)\cap \Sigma_t=C(K,t)$. But as $C(K,t)$ is compact, it easily follows that $C(K,t')$ is compact for all $|t'|\leq |t|$ and similarly for $J(K)\cap \Sigma_{t'}$. That $t^\infty(K)>0$ is proven as follows. As $K$ is compact, then, by Proposition~C.5 of Bullock~\cite{me}, $C(K,t)$ is compact for some $t>0$ and so $J^+(K)\cap \Sigma_t$ is compact by Proposition~\ref{J(K)compact}. It then follows that $t^\infty(K)>0$ and also that Statement~\textit{1} is true.  
If $t^\infty(K)<\infty$ and $J(K)\cap \Sigma_{t^{\infty}(K)}$ is compact then $C(K,t^\infty(K))$ is compact, as is $C(K,t^\infty(K)+\epsilon)$ for some $\epsilon>0$ (by Proposition~C.6 of Bullock~\cite{me}), and so also $J(K)\cap \Sigma_{t^{\infty}(K)+\epsilon}$ which contradicts the definition of $t^\infty(K)$. This proves Statement~\textit{2}. If $\Sigma$ is complete, then, for all $t$, as $C(K,t)$ is closed and bounded, so it's also compact by the Heine-Borel property. Statement~\textit{3} then follows from Proposition~\ref{J(K)compact}. If $p_n$ is a Cauchy sequence, then it is bounded and so contained in the compact set $C(K,t)$ for some t and so $p_n$ converges, which proves Statement~\textit{4}.
\end{proof}

\begin{prop}\label{anotherone}
Let $C(K,t)$ be compact in $\Sigma$, where $K$ is a compact subset of $\Sigma$ and $t\geq 0$, then $\{\frac{t}{2}\}\times C(K,\frac{t}{2})\subseteq D(\Sigma_0)$.
\end{prop}
\begin{proof}
This follows from Corollary~\ref{moreDSigma} with $s=\frac{t}{2}$.
\end{proof}

\begin{cor}
If $J(K)\cap \Sigma_t$ is compact, then $J(K)\cap\Sigma_{\frac{t}{2}}\subseteq D(\Sigma_0)$.
\end{cor}

\begin{proof}
This follows from Proposition~\ref{anotherone} and repeated use of Proposition~\ref{J(K)compact}.
\end{proof}
\begin{prop} For all $0\leq t_1\leq t_2$: 
\begin{enumerate}
\item $\pi(J(p)\cap\Sigma_{t_1})\subseteq \pi(J(p)\cap \Sigma_{t_2})$
\item $\pi(D(\Sigma_0)\cap\Sigma_{t_2})\subseteq \pi(D(\Sigma_0)\cap\Sigma_{t_1}),$
\end{enumerate}
where $\pi\colon\mathbb{R}\times \Sigma\rightarrow \Sigma$ is the map: $\pi(t,x)=x$.
\end{prop}

\begin{proof}
We can set w.l.o.g. $V=1$ since otherwise: 
\begin{align*}
\pi(J(p&)_{V^2dt^2-h}\cap\Sigma_{t_1})\\
&=\pi(J(p)_{dt^2-V^{-2}h}\cap\Sigma_{t_1})
\subseteq \pi(J(p)_{dt^2-V^{-2}h}\cap\Sigma_{t_2})
=\pi(J(p)_{V^2dt^2-h}\cap\Sigma_{t_2}),
\end{align*}
where our subscript notation highlights the dependence of $J(p)$ on the metric.

To prove Statement \textit{1}: If $q\in LHS$, then $\exists \gamma\colon[0,t_1]\rightarrow\mathbb{R}\times \Sigma, \gamma(t)=(t,\sigma(t))$, $|\dot{\sigma}(t)|\leq 1$, $\sigma(0)=p$, $\sigma(t_1)=q$. Let $\gamma'\colon[0,t_2]\rightarrow\mathbb{R}\times \Sigma$, $\gamma'(t)=(t, \sigma(t \frac{t_1}{t_2}))$, $|\dot{\gamma}'|^2=1-(\frac{t_1}{t_2})^2|\dot{\sigma}(t\frac{t_1}{t_2})|^2\leq 0$, $\gamma'(0)=(0,p)$, $\gamma'(t_2)=(t_2,q)$, so that $q\in RHS$. 

Statement \textit{2} follows from Proposition~\ref{DSigmaexplicit}.
\end{proof}
 
\begin{cor}
$J(p)\cap \Sigma_{t_1}\nsubseteq D(\Sigma_0)\Rightarrow J(p)\cap \Sigma_{t_2}\nsubseteq D(\Sigma_0)\;\forall 0\leq t_1\leq t_2$.
\end{cor}

\begin{prop}
If $t^\infty(K)< \infty$, then $J(K)\cap \Sigma_{t^\infty(K)/2}\nsubseteq D(\Sigma_0)$.
\end{prop}

\begin{proof}
Again w.l.o.g let $V=1$. We know via Propositions~\ref{tinfinity} and \ref{J(K)compact}, that: $\overline{B(p,t)}\subseteq \epsilon_p$ for all $t<t^\infty(K)$ and $p\in K$; $B(p,t^\infty(K))\subseteq \epsilon_p$ for all $p\in K$, and that there exists $p\in K$ such that $\overline{B(p,t^\infty(K))}\nsubseteq \epsilon_p$. Thus there exists $X_p\in T_p\Sigma\backslash \epsilon_p$ with $|X_p|= t^\infty(K)$. We hold that there must then exist a geodesic $\sigma\colon[0,1)\rightarrow \Sigma$ inextendible to 1 such that $\dot{\sigma}(0)=X_p$.

To show that this is true, let $\sigma:[0,a)\rightarrow \Sigma$ the maximal geodesic, starting at $p$ with $\dot{\sigma}(0)=X_p$. If $a>1$, then $X_p\in\epsilon_p$ by definition, which is however a contradiction. If $a<1$, then $\gamma$ be the geodesic through $p$ with $\dot{\gamma}(0)=aX_p$. By the rescaling Lemma, $\sigma$ being inextendible to $a$ implies that $\gamma$ is inextendible to $1$. So, by definition, $aX_p\nin\epsilon_p$. But $|aX_p|<|X_p|=t^\infty(K)$, which is a contradiction. 

Now that the existence of the geodesic $\sigma$ is proven, define $\sigma'\colon[0,t^\infty(K))\rightarrow \Sigma$ via: $\sigma'(s)=\sigma(\frac{s}{t^\infty(K)})$. It satisfies: $|\dot{\sigma}'(s)|=\frac{1}{t^\infty(K)}|\dot{\sigma}(\frac{s}{t^\infty(K)})|= 1$ and $\sigma'(0)=p$. So, from Proposition~\ref{causalfuture}, $x=(\frac{t^\infty(K)}{2},\sigma'(\frac{t^\infty(K)}{2}))\in J(K)\cap \Sigma_{t^\infty(K)/2}$. Now define $\alpha\colon(0,t^\infty(K)/2]\rightarrow \mathbb{R}\times \Sigma$, $\alpha(s)=(s, \sigma'(t^\infty(K)-s))$. Since $\sigma$ is inextendible to 1 then $\sigma'$ is inextendible to $t^\infty(K)$ and so $\alpha$ is past-inextendible to $0$. Clearly, $\alpha$ does not pass $\Sigma_0$ although $\alpha(t^\infty(K)/2)=(t^\infty(K)/2, \sigma'(t^\infty(K)/2))=x$. Since $\alpha$ is a future-pointing past-inextendible smooth causal curve passing $x$ but not $\Sigma_0$, then $x\nin D(\Sigma_0)$. Thus $x\in J(K)\cap \Sigma_{t^\infty(K)/2}\backslash D(\Sigma_0)$.
\end{proof}

\begin{cor}\label{moretinfinity} The following statements are true:
\begin{enumerate}
\item $J(K)\cap \Sigma_t\subseteq D(\Sigma_0)\;\forall 0\leq t< t^\infty(K)/2$.
\item $t^\infty(K)<\infty\Rightarrow J(K)\cap \Sigma_t\nsubseteq D(\Sigma_0)\;\forall  t\geq t^\infty(K)/2$.
\item $t_1(K)\colon=\sup\{t\colon J^+(K)\cap \Sigma_t\subseteq D(\Sigma_0)\}=t^\infty(K)/2$.
\end{enumerate}
\end{cor}

For the purposes of the following section, we continue these arguments to define an increasing sequence:
$$t_{n+1}(K)\colon=\sup\{t\colon J^+(K)\cap \Sigma_t\subseteq D(\Sigma_{t_n(K)})\},$$
where $t_0(K)=0$ and the resulting definition of $t_1(K)$ agrees with that used above. We are led to the following corollary:

\begin{cor} \label{tnsequence} The following statements are true:
\begin{enumerate}
\item $J(K)\cap \Sigma_t\subseteq D(\Sigma_{t_n(K)})\;\forall t_n(K)\leq t< (1-\frac{1}{2^n})t^\infty(K)$.
\item $t^\infty(K)<\infty\Rightarrow J(K)\cap \Sigma_t\nsubseteq D(\Sigma_{t_n(K)})\;\forall  t\geq (1-\frac{1}{2^n})t^\infty(K)$.
\item $t_n(K)=(1-\frac{1}{2^n})t^\infty(K)\nearrow t^\infty(K)$ as $n\rightarrow \infty$.
\end{enumerate}
\end{cor}

\section{Support of Wald Solutions}\label{sec:supportofwaldsolutions}
We now prove a result concerning the support of our ``Wald solutions''. It is in fact not true that given Cauchy data consisting of two test functions $(\phi_0, \dot{\phi}_0)$ then the support of the corresponding solution $\phi$ (w.r.t.\ some acceptable s.a.e.\ $A_E$ of $A$) as constructed in Theorem~\ref{waldsolutions}, is necessarily contained in $J(K)$, where $K=\supp \phi_0\cup \supp\dot{\phi}_0$ and $J(K)$ is as usual the union of the causal future and past of $K$: $J(K)=J^+(K)\cup J^-(K)$. A counterexample is given in Section~\ref{counterexample}.

It would however be natural to guess that up until a time at which data can pass to a possible edge, the support of $\phi$ is contained in $J(K)$. More precisely, if we define: $$t^\infty(K)= \sup\{t\geq 0\colon J^+(K)\cap \Sigma_t \text{ is compact in }\Sigma_t\}\in (0,\infty ],$$then we propose that $\supp\phi\cap [-t^\infty(K),t^\infty(K)]\times\Sigma\subseteq J(K).$ It was proven in Proposition~\ref{tinfinity} that  $t^\infty(K)>0$, so this is a non-trivial statement. At first sight it might appear that this result is trivial. Since $D(\Sigma_0)$ is a globally hyperbolic spacetime we know that $\supp \phi\cap D(\Sigma_0)\subseteq J(K)$ however this does not show that $\phi$ is zero in the shaded triangular region in Figure~\ref{grinch}. Thus this does not even prove that $\phi$ is compactly supported on $\Sigma_t$ for small $t$.

The proof we give shortly uses the uniqueness result of Lemma~\ref{uniquenessiii} and the sequence $t_n(K)$ constructed in the previous section. We shall define $\Psi\colon(-t_1(K),t_1(K))\times\Sigma\rightarrow \mathbb{R}$ to be equal to $\phi$ inside $J(K)$ and zero outside it. We shall show that $\Psi$ so defined is smooth, compactly supported on $\Sigma_t$ for $t\in (-t_1(K),t_1(K))$ and satisfies the Klein-Gordon equation in its domain. Thus  $[\Psi|_{\Sigma_t}]\in D(A)\subseteq D(A_E)$ and so $\Psi=\phi$ in the domain of $\Psi$ by Lemma~\ref{uniquenessiii}. By induction and the fact that $t_n(K)\nearrow t^\infty(K)$ the result then follows.

\begin{SCfigure}[50] 
\begin{tikzpicture}[scale=1] 
\draw[dash pattern=on 8pt off 8pt] (0,0) -- (0,5) node[label=right:$\mathbb{R}\times\Sigma$]{};
\draw[dash pattern=on 8pt off 8pt] (5,0) -- (5,5);
\draw(5,0) node[label=right:$\Sigma_0$]{};
\draw(5,1.8) node[label=right:$t_1(K)$]{};
\draw(5,2.7) node[label=right:$t_2(K)$]{};
\draw(5,3.6) node[label=right:$t^\infty(K)$]{};
\draw(5,3.15) node[label=right:$t_3(K)$]{};
\draw(2.75,0) node[label=below:$K$]{};
\draw(0,0)--(5,0);
\draw(1.25,1) node[label=below:$D^+(N_0)$]{};
\draw(4,1) node[label=below:$D^+(N_0)$]{};
\draw(4.5,1.9) node[label=below:$N_1$]{};
\draw(2.5,4) node[label=below:$J^+(K)$]{};
\draw (2.5,-0.1)--(2.5,0.1);
\draw (3,-0.1)--(3,0.1);
\draw (2.5,0)--(0,4.5);
\draw (0,1.8)--(1.5,1.8);
\draw[dash pattern=on 8pt off 8pt] (0,0)--(1.25,2.25);
\draw (3,0)--(5,3.6);
\draw[dash pattern=on 8pt off 8pt] (5,0)--(4,1.8);
\draw (4,1.8)--(5,1.8);
\draw[dash pattern=on 8pt off 8pt] (5,1.8)--(4.5,2.7);
\draw  (4.5,2.7)--(5,2.7);
\draw[dash pattern=on 8pt off 8pt] (5,2.7)--(4.75,3.15);
\draw (4.75,3.15)--(5,3.15);
\draw[dash pattern=on 8pt off 8pt,fill=gray!50](0,0)--(1,1.8)--(0,1.8)--cycle;
\end{tikzpicture}
\caption{The construction of $t^\infty(K)$ and $t_n(K)$ in for example $(M,g)=(\mathbb{R}\times (0,1), dt^2-dx^2)$.\vspace{6em}}
\label{grinch}
\end{SCfigure}
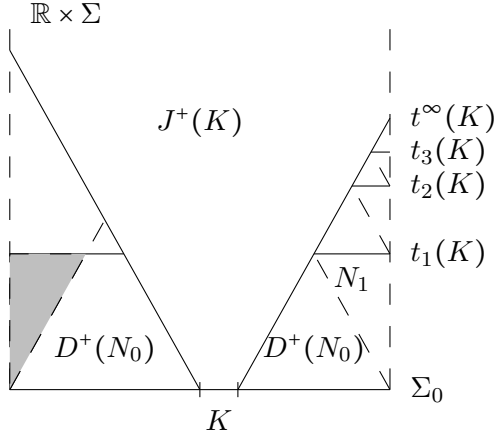
\begin{prop} 
Given $\phi_0, \dot{\phi}_0\in C^\infty_0(\Sigma)$ let $K=\supp \phi_0\cup \supp\dot{\phi}_0$. Define $t^\infty(K)$ as earlier. Let $\phi$ be the solution to the Klein-Gordon equation generated by some acceptable s.a.e.\ $A_E$ of $A$ and data $(\phi_0, \dot{\phi}_0)$ via Theorem~\ref{waldsolutions}. Then: 
\begin{enumerate}
\item $\text{If } t^\infty(K)=\infty\text{ then: }\supp\phi\subseteq J(K)$
\item $\text{If } t^\infty(K)<\infty\text{ then: }\supp\phi\cap [-t^\infty(K),t^\infty(K)]\times\Sigma\subseteq J(K)$
\end{enumerate}
\end{prop}

\begin{proof}
If $t^\infty(K)=\infty$, then, by Proposition~\ref{tinfinity}, $(\Sigma, V^{-2}h)$ is a complete Riemannian manifold and so $M$ is globally hyperbolic by Lemma~\ref{staticandglobhyp} and $\supp\phi\subseteq J(K)$ follows from Theorem~\ref{globhypsolution1}. If $t^\infty(K)<\infty$, construct a strictly increasing sequence $(t_n(K))_{n\geq 0}$ inductively as follows. Let $t_0(K)=0$ and $t_{n+1}(K)\colon=\sup\{t\colon J^+(K)\cap \Sigma_t\subseteq D(\Sigma_{t_n(K)})\}$. From Corollary~\ref{tnsequence}, we know that $t_n(K)\nearrow t^\infty(K)$. Define $\Psi\colon(-t_1(K),t_1(K))\times\Sigma\rightarrow \mathbb{R}$ as:
$$\Psi(x)=\left\{\begin{array}{ll}\phi(x),  & \text{ for } x\in (-t_1(K),t_1(K))\times\Sigma\cap J(K)\\
0, &\text{ otherwise.}\end{array}\right.$$
The first problem is to show that the function $\Psi$ so defined is smooth. We do this by finding for each $x\in (-t_1(K),t_1(K))\times\Sigma$ an open neighbourhood $U$ s.t. $\Psi$ either equals $\phi$ on $U$ or is zero on $U$. If $x\in (-t_1(K),t_1(K))\times\Sigma\cap D(\Sigma_0)=U$, which is an open neighbourhood (since $D(\Sigma_0)$ is open by Proposition~\ref{DSigmaisglobhyp}), then $\Psi=\phi$ on $U$. This is because if $y\in U$ then either $y\in J(K)$ and so $\Psi(y)=\phi(y)$ by definition, or $y\in D(\Sigma_0)\backslash J(K)=D(\Sigma_0\backslash K)$ and $\Psi(y)=0=\phi(y)$ (by the uniqueness of solutions to the Klein-Gordon equation on the globally hyperbolic spacetime $D(\Sigma_0\backslash K)$ (Theorem~\ref{globhypsolution1}), where $\Sigma_0\backslash K$ is an acausal topological hypersurface and so $D(\Sigma_0\backslash K)$ is an open set in $M$ and a globally hyperbolic spacetime by Proposition~\ref{DSigmaisglobhyp}). If $x\in (-t_1(K),t_1(K))\times\Sigma\backslash D(\Sigma_0)\subseteq (-t_1(K),t_1(K))\times\Sigma\backslash J(K)=\colon U$ (from Corollary~\ref{moretinfinity}, Statement~\textit{1}), then $\Psi=0$ on $U$ by definition.

Thus $\Psi\in C^\infty((-t_1(K),t_1(K))\times\Sigma)$ and $[\Psi_t],[\partial_t\Psi_t]\in [C_0^{\infty}(\Sigma)]=D(A)\subseteq D(A_E)$, for all $t\in[0,t_1(K))$. We also have $\Psi|_{\Sigma_0}=\phi_0$ and $\partial_t \Psi|_{\Sigma_0}=\dot{\phi}_0$. Since $\Psi$ is locally either equal to $\phi$, or zero, both being solutions of the Klein-Gordon equation, then so is $\Psi$, that is $(\largesquare_g+m^2)\Psi=0$ on $(-t_1(K),t_1(K))\times\Sigma$. Moreover, by definition $\Psi=0$ on $[0,t_1(K))\times\Sigma\backslash J(K)$. By uniqueness of the Wald solution (Lemma~\ref{uniquenessiii}), then: $\phi=\Psi\text{ in }[0,t_1(K))\times \Sigma$. Therefore, $\phi=0$ on $[0,t_1(K))\times\Sigma\backslash J(K)$. But since $\phi$ is smooth, then also $\partial_t \phi=0$ on $[0,t_1)\times\Sigma\backslash J(K)$. In particular, $\phi=\partial_t \phi=0$ on $\Sigma_{t_1(K)}\backslash J(K)=N_1$.

Using the constructed sequence $(t_n(K))_{n\geq 0}$, we prove the proposition by induction. Our inductive hypothesis $P(n)$ is as follows: $P(n)\colon\; \supp\phi\cap[0,t_n(K)]\times\Sigma\subseteq J(K)$. We have already proven the statement for $n=1$. If $P(n)$ is true, by smoothness $\phi,\partial_t \phi$ are zero on $N_n=\Sigma_{t_n(K)}\backslash J(K)$. Now, as before, define: $\Psi\colon[t_n(K),t_{n+1}(K))\times\Sigma\rightarrow \mathbb{R}$ as:
$$\Psi(x)=\left\{\begin{array}{ll}\phi(x), &  \text{ for } x\in [t_n(K),t_{n+1}(K))\times\Sigma\cap J(K)\\
0, &\text{ otherwise.}\end{array}\right.$$
Similarly to the previous argument $\Psi\in C^\infty([t_n(K),t_{n+1}(K))\times\Sigma)$ as a manifold with boundary. Also:
$$[\Psi_{\Sigma_t}]\in [C_0^{\infty}(\Sigma)]=D(A)\subseteq D(A_E)\;\;\forall t\in[t_n(K),t_{n+1}(K)),$$
$$\Psi|_{\Sigma_{t_n(K)}}=\phi_{t_n(K)}\text{ and }\partial_t \Psi|_{\Sigma_0}=\dot{\phi}_{t_n(K)}.$$
By the uniqueness theorem (Lemma~\ref{uniquenessiii}), $\phi=\Psi$ in $[t_n(K),t_{n+1}(K))\times \Sigma$. Thus $\phi=0$ on $[t_n(K),t_{n+1}(K))\times\Sigma\backslash J(K)$. But since $\phi$ is smooth, then also $\partial_t \phi=0$ on $[t_n(K),t_{n+1}(K))\times\Sigma\backslash J(K)$. In particular then, $\phi=0$ on $\Sigma_{t_{n+1}(K)}\backslash J(K)=N_1$ and $P(n+1)$ is proven. Hence $\supp\phi\cap[0,t_n(K)]\times\Sigma\subseteq J(K) \text{ for all } n$. But as $t_n(K)\nearrow t^\infty(K)$, then $\supp\phi\cap[0,t^\infty(K))\times\Sigma\subseteq J(K)$, and by continuity: $$\supp\phi\cap[0,t^\infty(K)]\times\Sigma\subseteq J(K).$$
Finally, since the spacetime is symmetric around $\Sigma_0$, we have: $$\supp\phi\cap[-t^\infty(K),t^\infty(K)]\times\Sigma\subseteq J(K).$$ 
\end{proof}

\section{Energy form on the Space of Solutions}\label{sec:bilinearform}

In Sections~\ref{sec:bilinearform} to \ref{symmetries}, we shall prove the existence of certain structures on the space of solutions $S_E$ (Definition~\ref{spaceofsolutions}), corresponding to a particular acceptable s.a.e.\ $A_E$. Specifically, we shall show the existence of an energy form, a symplectic form and certain symmetries: time translation and time-reversal. These were all conditions placed on the dynamics in the paper by Wald and Ishibashi~\cite{c}. It is important for us to show that these conditions are in fact necessary, even in our extended case of dynamics generated by an acceptable s.a.e.\ $A_E$. In this section we show that there is a natural bilinear symmetric form $E$ on our constructed space of solutions $S_E$ to the Klein-Gordon equation. In general, it is not a norm. However, if our choice of acceptable self-adjoint extension $A_E$ is positive and zero is not an eigenvalue, then $E$ is a norm on $S_E$.

Given two pairs of smooth Cauchy data: $(\phi_0,\dot{\phi}_0),(\phi_0',\dot{\phi}_0')\in \chi_E^2\subseteq C^\infty(\Sigma)^2$ then we have by the existence of Wald solutions (Theorem~\ref{waldsolutions}) two corresponding solutions $\phi, \phi'$ to the Klein Gordon equation on our spacetime. For each time $t\in \mathbb{R}$ we define the energy at time $t$ to be:
$$E(\phi,\phi')(t)=\langle \dot{\phi}_t,\dot{\phi}_t'\rangle_{\Sigma_t} +\langle \phi_t,A_E \phi_t'\rangle_{\Sigma_t}$$ 
Our task is to show that $E(\phi,\phi')$ is in fact independent of time. Remember that: 
\begin{align*}
\phi_t&=C(t,A_E)\phi_0 + S(t,A_E)\dot{\phi}_0\\
\dot{\phi}_t&=-A_E S(t,A_E)\phi_0 +C(t,A_E)\dot{\phi}_0
\end{align*}
Thus, for all $t\in \mathbb{R}$:
\begin{align*}
E(\phi,\phi')(t)&=\langle -A_E S(t,A_E)\phi_0 +C(t,A_E)\dot{\phi}_0,-A_E S(t,A_E)\phi_0'+C(t,A_E)\dot{\phi}_0'\rangle \\
&\;\;\;\;\;\;\;+\langle C(t,A_E)\phi_0 + S(t,A_E)\dot{\phi}_0,A_E C(t,A_E)\phi_0' + A_E S(t,A_E)\dot{\phi}_0' \rangle\\
&=\langle A_E S(t,A_E)\phi_0,A_E S(t,A_E)\phi_0'\rangle -\langle A_E S(t,A_E)\phi_0,C(t,A_E)\dot{\phi}_0'\rangle \\
&\;\;\;\;\;\;\;-\langle C(t,A_E)\dot{\phi}_0,A_E S(t,A_E)\phi_0'\rangle +\langle C(t,A_E)\dot{\phi}_0,C(t,A_E)\dot{\phi}_0'\rangle\\
&\;\;\;\;\;\;\;+\langle C(t,A_E)\phi_0 ,A_E C(t,A_E)\phi_0'\rangle + \langle C(t,A_E)\phi_0 ,A_E S(t,A_E)\dot{\phi}_0' \rangle\\
&\;\;\;\;\;\;\;+\langle S(t,A_E)\dot{\phi}_0,A_E C(t,A_E)\phi_0'\rangle+ \langle S(t,A_E)\dot{\phi}_0,A_E S(t,A_E)\dot{\phi}_0' \rangle\\
&=\langle \phi_0,A_E(A_E S(t,A_E)^2+C(t,A_E)^2)\phi_0'\rangle \\
&\;\;\;\;\;\;\;+ \langle \dot{\phi}_0,(A_E S(t,A_E)^2+C(t,A_E)^2)\dot{\phi}_0'\rangle\\
&=\langle \phi_0,A_E\phi_0'\rangle + \langle \dot{\phi}_0,\dot{\phi}_0'\rangle\\
&=E(\phi,\phi')(0),
\end{align*}
where we have used the following identity: $A_E S(t,A_E)^2+C(t,A_E)^2=\mathbb{I}$ on $[\chi_E]$. Hence $E(t)$ has the same value at all times. Using the linear isomorphism $\Psi\colon \chi_E\times \chi_E \rightarrow S_E$ between $\chi_E^2$ and the space of solutions $S_E$ defined in Proposition~\ref{spaceofsolutions} then $E$ defined above is a bilinear symmetric form on $S_E$ (the symmetry of $E$ follows easily since as $A_E$ is self-adjoint it is certainly symmetric) and is called the \textbf{energy form}. 

\section{The Symplectic Form on the Space of Solutions}\label{sec:sympform}

Similarly to the previous section, we show that there exists a natural symplectic form on the real vector space of our space of solutions $S_E$. Given two pairs of smooth Cauchy data $(\phi_0,\dot{\phi}_0),(\phi_0',\dot{\phi}_0')\in \chi_E^2\subseteq C^\infty(\Sigma)^2$, then we have by the existence of Wald solutions (Theorem~\ref{waldsolutions}) two corresponding solutions $\phi, \phi'$ to the Klein Gordon equation on our spacetime. For each time $t\in \mathbb{R}$ we define the \textbf{symplectic form} at time $t$ to be:
$$\sigma_E(\phi,\phi')(t)=\langle \phi_t,\dot{\phi}_t'\rangle -\langle \dot{\phi}_t,\phi_t'\rangle$$
We show again  that this form is independent of time. For all $t \in \mathbb{R}$:
\begin{align*}
\sigma_E(\phi,\phi')(t)&=\langle C(t,A_E)\phi_0 + S(t,A_E)\dot{\phi}_0,-A_E S(t,A_E)\phi_0' +C(t,A_E)\dot{\phi}_0'\rangle\\
&\;\;\;\;\;+\langle A_E S(t,A_E)\phi_0 -C(t,A_E)\dot{\phi}_0,C(t,A_E)\phi_0' + S(t,A_E)\dot{\phi}_0'\rangle\\
&=-\langle C(t,A_E)\phi_0,A_E S(t,A_E)\phi_0'\rangle +\langle C(t,A_E)\phi_0,C(t,A_E)\dot{\phi}_0'\rangle\\
&\;\;\;\;\;-\langle S(t,A_E)\dot{\phi}_0,A_E S(t,A_E)\phi_0'\rangle + \langle S(t,A_E)\dot{\phi}_0,C(t,A_E)\dot{\phi}_0'\rangle\\
&\;\;\;\;\;+\langle A_E S(t,A_E)\phi_0, C(t,A_E)\phi_0'\rangle +\langle A_E S(t,A_E)\phi_0,S(t,A_E)\dot{\phi}_0'\rangle\\
&\;\;\;\;\;-\langle C(t,A_E)\dot{\phi}_0,C(t,A_E)\phi_0'\rangle -\langle C(t,A_E)\dot{\phi}_0,S(t,A_E)\dot{\phi}_0'\rangle\\
&=\langle \phi_0,(A_E S(t,A_E)^2+C(t,A_E)^2)\dot{\phi}_0'\rangle\\
&\;\;\;-\langle \dot{\phi}_0,(A_E S(t,A_E)^2+C(t,A_E)^2)\phi_0'\rangle\\
&=\langle \phi_0,\dot{\phi}_0'\rangle-\langle \dot{\phi}_0,\phi_0'\rangle\\
&=\sigma_E(\phi,\phi')(0)
\end{align*}
Here, we have again made use of the identity $A_E S(t,A_E)^2+C(t,A_E)^2=\mathbb{I}$ on $[\chi_E]$. Thus we have a map $\sigma_E\colon S_E\times S_E\rightarrow \mathbb{R}$, where $S_E$ is the real vector space of solutions. It is clearly bilinear, antisymmetric and also weakly nondegenerate, since if $\phi\in S_E$ is non-zero then (by uniqueness) $(\phi_0,\dot{\phi}_0)\neq (0,0)\in \chi_E\times\chi_E$. Consequently, let $\phi'=\Psi(-\dot{\phi}_0,\phi_0)$. Then, $\sigma_E(\phi,\phi')=||\phi_0||^2+||\dot{\phi}_0||^2>0$ as either $\phi_0$ or $\dot{\phi}_0$ is non-zero and so has non-zero norm (as both are continuous). Thus, $(S_E,\sigma_E)$ is a real symplectic space.

\section{Symmetries}\label{symmetries}

In this section, we derive some symmetries satisfied by the linear isomorphism $\Psi:\chi_E\times\chi_E\rightarrow S_E$ defined in Proposition~\ref{spaceofsolutionsiso}. Consider the maps $T_t,P:C^\infty(M)\rightarrow C^\infty(M)$ given by:
\begin{align*}
(T_tF)(s,x)&=F(s-t,x)\\
(PF)(s,x)&=F(-s,x)
\end{align*}
\begin{prop}
Given a standard static spacetime and the linear operator $A$ defined as usual on the Hilbert space $L^2(\Sigma, d\vol_h)$ then for any acceptable s.a.e.\ $A_E$ of $A$. The maps $T_t$ and $P$ satisfy: $T_t, P\colon\; S_E\rightarrow S_E$.
Then letting $\phi_t=\Psi(\phi_0,\dot{\phi}_0)|_{\Sigma_t}$ and 
$\dot{\phi}_t=\partial_t\Psi(\phi_0,\dot{\phi}_0)|_{\Sigma_t}$ we have:
\begin{align*}
\Psi(\phi_t, \dot{\phi}_t)&=T_{-t}[\Psi(\phi_0,\dot{\phi}_0)]\\
\Psi(\dot{\phi}_0,-A_E\phi_0)&=\frac{\partial}{\partial t}[\Psi(\phi_0, \dot{\phi}_0)]\\
\Psi(\phi_0,-\dot{\phi}_0)&=P[\Psi(\phi_0,\dot{\phi}_0)]
\end{align*}
In particular, this also proves that $\frac{\partial}{\partial t}\colon\;\chi_E\rightarrow\chi_E$. Additionally, for all $\Psi_1,\Psi_2\in S_E$:
\begin{align*}
E(T_t \Psi_1, T_t\Psi_2)&=E(\Psi_1,\Psi_2)\\
E(P \Psi_1, P\Psi_2)&=E(\Psi_1,\Psi_2)\\
\sigma_E(T_t \Psi_1, T_t\Psi_2)&=\sigma_E(\Psi_1,\Psi_2)\\
\sigma_E(P \Psi_1, P\Psi_2)&=-\sigma_E(\Psi_1,\Psi_2)
\end{align*}
\end{prop}
(Note that the first five properties correspond to Assumptions 2(i), 2(ii), 3(i) and 3(ii) in Wald and Ishibashi \cite{c}.)
\begin{proof}
\begin{align*}
\Psi(\phi_t, \dot{\phi}_t)(s,x)&=\left[\begin{array}{l}C(s,A_E)(C(t,A_E)\phi_0+S(t,A_E)\dot{\phi}_0)\\
+S(s,A_E)(-A_E S(t,A_E)\phi_0+C(t,A_E)\dot{\phi}_0)\end{array}\right](x)\\
&=[C(s+t,A_E)\phi_0+S(s+t,A_E)\dot{\phi}_0](x)\\
&=\Psi(\phi_0, \dot{\phi}_0)(t+s,x)\\
&=T_{-t}(\Psi(\phi_0,\dot{\phi}_0))(s,x)
\end{align*}
$$\Psi(\dot{\phi}_0,-A_E\phi_0)(t,s)=\left[C(t,A_E)\dot{\phi}_0+S(t,A_E)(-\phi_0))\right](x)
=\dot{\phi}_t(x)
=\frac{\partial}{\partial t}[\Psi(\phi_0,\dot{\phi}_0)](t,x)$$
\begin{align*}
\Psi(\phi_0,-\dot{\phi}_0)(t,x)&=[C(t,A_E)\phi_0+S(t,A_E)(-\dot{\phi}_0))](x)\\
&=[C(-t,A_E)\phi_0+S(-t,A_E)\dot{\phi}_0)](x)\\
&=\Psi(\phi_0,\dot{\phi}_0)(-t,x)\\
&=P(\Psi(\phi_0,\dot{\phi}_0))(t,x)
\end{align*}
The remaining properties are easily proven from the time independence of $E(\Psi_1,\Psi_2)(t)$ and $\sigma(\Psi_1,\Psi_2)(t)$ (Sections~\ref{sec:bilinearform} and~\ref{sec:sympform}). 
\end{proof}

\section{Examples}\label{examples}
In this section we shall discuss a few simple examples of standard static spacetimes. In all the examples we examine we shall let $V=1$ for simplicity. Thus the spacetime $(M,g)=(\mathbb{R}\times \Sigma, dt^2-h)$ and the solutions to the Cauchy problem of the Klein-Gordon equation constructed in this paper for each of these spacetimes will be indexed by the acceptable s.a.e.s $A_E$ of the symmetric linear operator $A$ on $L^2(\Sigma, d\vol_h)$ generated by the partial differential operator (also labelled by) $A=-\divergence_h \grad_h$, minus the Laplace-Beltrami operator, with $D(A)=[C_0^\infty(\Sigma)]$.

We shall also only consider the case of the solving the Klein-Gordon case for complex-valued data and so we only consider complex Hilbert spaces. Note that it's only on complex Hilbert spaces that we can define the deficiency spaces $H^\pm$ of a densely defined operator $A$ as $H^\pm:=\ker(A^*\mp i)$. We note the following theorem (see Theorems 83.1 and 85.1 in Akhiezer and Glazman \cite{x}):
\begin{thm}\label{glazman}
Let $A$ be a positive symmetric linear operator with equal and finite deficiency indices, that is, denoting $n^{\pm}\colon=\dim \ker(A^*\mp i)$, we have $n^+=n^-=n < \infty$. Then every s.a.e.\ $A_E$ of $A$ is  bounded-below. Furthermore, every s.a.e.\ $A_E$ has the same continuous spectrum as $A$, each of the s.a.e.s has only a finite number of negative eigenvalues and the sum of the multiplicities of the negative eigenvalues of any particular s.a.e.\ $A_E$ is not greater than $n$.
\end{thm}
Thus if $A$ has finite deficiency indices then in particular every s.a.e.\ $A_E$ of $A$ is acceptable. In all the following examples the deficiency indices are finite and are equal to 0, 1 or 2.

\subsection{Self-Adjoint Extensions of minus the Laplacian on $S^1$}\label{circle}
In our first example we let $\Sigma=S^1$. We equip $S^1$ with its (unique) differential structure, the Riemannian metric induced from that on $\mathbb{R}^2$ and the induced smooth measure from this metric. Since $S^1$ is compact in its topology induced from $\mathbb{R}^2$, then it is also compact in its topology induced from the Riemannian metric $h$ (Theorem~\ref{Riemmetric}), so it is also complete in this metric and so complete as a Riemannian manifold by the Hopf-Rinow Theorem (See e.g.\ Theorem 6.13 Lee \cite{p}). Thus the linear operator $A$ given by $D(A)=[C_0^\infty(S^1)]=[C^\infty(S^1)]$, $A([\phi])=-[\phi'']\text{ for } \phi\in C^\infty(S^1)$ is essentially self-adjoint by Theorem~\ref{esa}. Thus $\overline{A}=A^*$ is the unique s.a.e.\ of $A$ and  
$$D(\overline{A})=W^{2,2}(S^1)=\{\phi\in L^2(S^1) \text{ s.t. } \phi',\;\phi''\in L^2(S^1)\}.$$
Note the Sobolev space $W^{2,2}(S^1)$ is defined in Appendix~D.3 of Bullock~\cite{me}, where we are implicitly adopting the standard Riemannian metric on $S^1$ as on all the manifolds in Section~\ref{examples}.

The spectrum of $\overline{A}$ is shown in the appendix to be:
$$\sigma(\overline{A})=\sigma_{disc}(\overline{A})=\{n^2\colon\; n\in \mathbb{N}_0\}.$$
If we identify $S^1\backslash \{1\}$ with $(0,2\pi)$ by the chart: $\phi\colon\; U=S^1\backslash \{1\}\rightarrow (0,2\pi)$, $\phi^{-1}(\theta)=\exp i\theta$, then define the function $g\colon U\times U\times \mathbb{C}\backslash\{n^2\colon\; n\in\mathbb{N}\}\rightarrow \mathbb{C}$ by:
$$g(\theta, \phi; \lambda)=\frac{i}{2\sqrt{\lambda}}\left[ \exp i\sqrt{\lambda}|\theta-\phi|+\frac{2\cos \sqrt{\lambda}(\theta-\phi)}{\exp (-2\pi i\sqrt{\lambda})-1}\right].$$
As $\{1\}\subseteq S^1$ is clearly null, $h$ generates a well-defined integral kernel.

The Green's function for $\lambda\in \rho(\overline{A})=\mathbb{C}\backslash \{n^2\colon \;n\in \mathbb{N}_0\}$ is given by $g(\cdot,\cdot,\lambda)$, which does not depend on the choice of square root of $\lambda$ used to define it.

\subsection{Self-Adjoint Extensions of minus the Laplacian on $(0,\infty)$}\label{halfline}
In the remaining cases the domains $D(A_E)$ of the s.a.e.s of $A$ shall be given by conditions placed on the domain of the adjoint of $A$, that is $D(A^*)$. Since $A\leq A_E\leq A^*$ then all the s.a.e.s of $A$ are restrictions of $A^*$ to their domain $D(A_E)$. The conditions placed on the domain will be in terms of ``trace maps''. The derivation of these maps is to be found in Lions and Magenes \cite{t}. 

\begin{thm}\label{boundaryvalues}
Let $\Omega$ be an open interval of $\mathbb{R}$. Consider the linear maps: 
\begin{align*}
&\rho\colon\;C_0^\infty(\overline{\Omega})\rightarrow \mathbb{C}^{|\partial \Omega|},\;\; \phi\mapsto \phi|_{\partial \Omega}\\
&\tau\colon\;C_0^\infty(\overline{\Omega})\rightarrow \mathbb{C}^{|\partial \Omega|},\;\; \phi\mapsto \phi'|_{\partial \Omega}.
\end{align*}
These maps extend by continuity to a unique continuous maps $$\rho,\tau \colon\;W^{2,2}(\Omega)\rightarrow \mathbb{C}^{|\partial \Omega|}.$$ Letting $\Phi=(\rho,\tau)$, then $\Psi$ is linear and surjective. Additionally: $W_0^{2,2}(\Omega)=ker \Psi=\{\phi\in W^{2,2}(\Omega)\colon\; \rho(\phi)=\tau(\phi)=0\}$.
\end{thm}

Note that $C_0^\infty(\overline{\Omega})$ is defined as the space of smooth functions on the smooth manifold with boundary $\overline{\Omega}$ which are of compact support. If $\overline{\Omega}$ is compact then clearly $C_0^\infty(\overline{\Omega})=C^\infty(\overline{\Omega})$. Note also that we define: $$W_0^{2,2}(\Omega)\colon=\overline{[C_0^\infty(\Omega)]}^{W^{2,2}(\Omega)},$$the closure of $[C_0^\infty(\Omega)]$ in the Sobolev norm on $W^{2,2}(\Omega)$. 

If $\Sigma=(0,\infty)$ then the s.a.e.s of $A$ are indexed by $\alpha\in (-\pi/2,\pi/2]$, denoted $A_\alpha$. Their domains are given by:
$$D(A_\alpha)=\{\phi\in W^{2,2}(0,\infty)\text{ s.t. } \cos \alpha \;\rho(\phi)=\sin \alpha \;\tau(\phi)\}$$
The spectra of these s.a.e.s are given by the following: 
$$\sigma(A_\alpha)=\left\{\begin{array}{ll}[0,\infty)&\text{ for } \alpha\in[0,\pi/2]\\
\left[0,\infty)\cup\{-\cot^2\alpha\}\right.&\text{ for } \alpha\in(-\pi/2,0)\end{array}\right.$$
The pure point spectrum $\sigma_{pp}(A_\alpha)$, continuous spectrum $\sigma_{cont}(A_\alpha)$ and discrete spectrum\linebreak $\sigma_{disc}(A_\alpha)$ of the operator $A_\alpha$ are given by the following statements:
\begin{enumerate}
\item $\sigma_{cont}(A_\alpha)=[0,\infty)$ for all $\alpha$.
\item If $\alpha\in [0,\frac{\pi}{2}]$: $\sigma(A_\alpha)=\sigma_{cont}(A_\alpha)=[0,\infty)$.
\item If $\alpha\in (-\frac{\pi}{2},0)$: $\sigma_{pp}(A_\alpha)=\sigma_{disc}(A_\alpha)=\{-\cot^2\alpha\}$.
\end{enumerate}
Thus for $\alpha\in [0,\frac{\pi}{2}]$, $A_\alpha$ is a positive s.a.e.\ of $A$ and for $\alpha\in (-\frac{\pi}{2},0)$, $A_\alpha$ is not positive but it is bounded-below. Thus all the s.a.e.s of $A$ are acceptable according to Definition~\ref{acceptablesae}.

For completeness we also now give the Green's function for each s.a.e.\ $A_\alpha$, that is $g\colon(0,\infty)\times(0,\infty)\times
\rho(A_\alpha)\rightarrow \mathbb{C}$ given by:
$$g(x,\xi,\lambda)=A[\cos\alpha\sin(\sqrt\lambda
x_<)+\sqrt\lambda\sin\alpha\cos(\sqrt\lambda x_<)]\exp(i\sqrt{\lambda}x_>),$$
where $A=[\sqrt\lambda (\cos\alpha-i\sqrt\lambda\sin\alpha)]^{-1}$,
$x_>=\max\{x,\xi\}$, $x_<=\min\{x,\xi\}$ and $\sqrt\lambda=a+bi, b>0$ is defined as the
unique square root of $\lambda$ in the upper-half plane, possible since $\lambda\nin
[0,\infty)$. These statements are proven in Appendix~G of Bullock~\cite{me}.

\subsection{Self-Adjoint Extensions of minus the Laplacian on $(0,a)$}\label{interval}
Now consider the case where $\Sigma=(0,a)$. Again, the domains of the s.a.e.s of $A$ are given in terms of the elements of $W^{2,2}(0,a)$ satisfying conditions placed on them via the trace map. The set of all s.a.e.s of $A$ are given by the Dirichlet extension and two groups of extensions which we shall describe presently. The first group shall also contain the Neumann extension. 

We shall give a very brief explanation of the origin of this classification of the self-adjoint extensions of the $A$. We refer the reader to Posilicano \cite{v}. The set of s.a.e.s of $A$ is indexed by pairs:
$$\left\{(\Pi,\Theta)\colon \begin{array}{l}\Pi \text{ is an orthogonal projection operator on the Hilbert space $\mathbb{C}^2$}\\ 
\Theta \text{ is a bounded s.a. linear operator on the Hilbert space $Im(\Pi)$} 
\end{array}\right\}.$$  
Given the pair $(\Pi,\Theta)$ the s.a.e.\ $A_{\Pi,\Theta}$ is then given by:
$$D(A_{\Pi,\Theta})=\left\{\phi\in W^{2,2}(0,a)\colon\;\rho\phi\in Im(\Pi),\;\Pi\tau\phi=\Theta\rho\phi\right\},$$
where:
\begin{align*}
&\rho\colon W^{2,2}(0,a)\rightarrow\mathbb{C}^2,\;\; \rho(\phi)=\left(\begin{array}{l}\phi(0)\\ \phi(a)\end{array}\right)\\
&\tau\colon W^{2,2}(0,a)\rightarrow \mathbb{C}^2, \;\;\tau(\phi)=\left(\begin{array}{l}\phi'(0)\\ -\phi'(a)\end{array}\right).
\end{align*}
Here, we are implicitly using Theorem~\ref{boundaryvalues}. Note that $\tau$ evaluates the inward-pointing derivative at the boundary, hence the sign. 

There are three natural collections of s.a.e.s of $A$ according as $rank(\Pi)=0,1\text{ or }2$. Picking $rank(\Pi)=0$, then we have $\Pi=0$ and $\Theta=0$. This is the Dirichlet s.a.e.\ $A_D$ (which we may also call the \textbf{s.a.e.\ of the zeroth kind}), defined by the following domain:  
$$D(A_D)=\{\phi\in W^{2,2}(0,a)\text{ s.t. } \phi(0)=\phi(1)=0\}.$$

Picking $rank(\Pi)=2$, we obtain the next collection of s.a.e.s, which we shall call the \textbf{s.a.e.s of the first kind}, in agreement with the language of Posilicano~\cite{v}. They are obtained by setting $\Pi=\mathbb{I}$. Then  $Im\Pi=\mathbb{C}^2$ and let $\Theta$ be defined by a self-adjoint complex $2\times 2$ matrix $$\theta=\left(\begin{array}{cc}
\theta_{11} & \theta_{12}\\
\overline{\theta}_{12} & \theta_{22}
\end{array}\right),$$where $\theta_{11}, \theta_{22}\in \mathbb{R}, \theta_{12}\in \mathbb{C}$. The domain of the extension $D(A_\theta)$ is defined as those elements $\phi\in W^{2,2}(0,a)$ such that $\Pi\tau\phi=\Theta\rho\phi$, that is:
$$\left(\begin{array}{c}\phi'(0)\\-\phi'(a)\end{array}\right)=\left(\begin{array}{cc}
\theta_{11} & \theta_{12}\\
\overline{\theta}_{12} & \theta_{22}
\end{array}\right)\left(\begin{array}{l}\phi(0)\\\phi(a)\end{array}\right),$$
or:
$$D(A_\theta)=\left\{\phi\in W^{2,2}(0,a)\text{ s.t.: }
\begin{array}{l}
\theta_{11}\phi(0)-\phi'(0)+\theta_{12}\phi(a)=0\\
\overline{\theta_{12}}\phi(0)+\theta_{22}\phi(a)+\phi'(a)=0
\end{array}\right\}.$$
Note that letting $\theta_{11}=\theta_{22}=\theta_{12}=0$ we obtain the Neumann extension $A_N=A_0$:
$$D(A_N)=\{\phi\in W^{2,2}(0,a)\text{ s.t. } \phi'(0)=\phi'(a)=0\}.$$

Picking $rank(\Pi)=1$, we obtain the last collection of s.a.e.s, which we shall call the \textbf{s.a.e.s of the second kind}. They are obtained by setting $$\Pi=w\otimes w=\left(\begin{array}{ll}|w_1|^2 & w_1\overline{w_2}\\
\overline{w_1}w_2&|w_2|^2\end{array}\right),$$
where $w=(w_1,w_2)\in\mathbb{C}^2$ is a unit vector, spanning a one-dimensional subspace in $\mathbb{C}^2$ and defining $\Pi$ by orthogonal projection onto this subspace. We set $\Theta$ to be defined as multiplication by $\theta\in\mathbb{R}$. These s.a.e.s are then indexed by triples: $\{(w_1,w_2,\theta)\colon\;w_1,w_2\in \mathbb{C} \text{ s.t. } |w_1|^2+|w_2|^2=1 \text{ and } \theta\in \mathbb{R}\}$. The domain of the extension $D(A_{w_1w_2\theta})$ is then those elements $\phi\in W^{2,2}(0,a)$ such that:
$$\rho(\phi)=\left(\begin{array}{l}\phi(0)\\\phi(a)\end{array}\right)\in Im\Pi=\left\langle\left(\begin{array}{l}w_1\\w_2\end{array}\right)\right\rangle$$
\begin{align*}
\left(\begin{array}{l}\theta\phi(0)\\\theta\phi(a)\end{array}\right)=\theta\rho\phi&=\Pi\tau\phi\\
&=\Pi\left(\begin{array}{c}\phi'(0)\\-\phi'(a)\end{array}\right)\\
&=\left(\begin{array}{ll}|w_1|^2 & w_1\overline{w_2}\\
\overline{w_1}w_2&|w_2|^2\end{array}\right)\left(\begin{array}{c}\phi'(0)\\-\phi'(a)\end{array}\right)
\end{align*}
Then, from the first condition: $w_2\phi(0)=w_1\phi(a)$. And from the second: 
\begin{align*}
\overline{w_1}\theta\phi(0)+\overline{w_2}\theta\phi(a)&=\overline{w_1}|w_1|^2\phi'(0)-\overline{w_2}|w_1|^2\phi'(a)+\overline{w_1}|w_2|^2\phi'(0)-\overline{w_2}|w_2|^2\phi'(a)\\
&=\overline{w_1}\phi'(0)-\overline{w_2}\phi'(a)
\end{align*}
Thus, $$D(A_{w_1w_2\theta})=\left\{\phi\in W^{2,2}(0,a)\text{ s.t.: }
\begin{array}{l}
w_2\phi(0)-w_1\phi(a)=0\\
\overline{w_1}(\theta\phi(0)-\phi'(0))+\overline{w_2}(\theta\phi(a)+\phi'(a))=0
\end{array}\right\}.$$
\begin{rem}
Note that replacing $w_i$ with $w_i e^{i\phi}$ for $i=1,2$ then we obtain identical boundary conditions and so the same s.a.e.\ of $A$. Clearly this is because both choices yield the same 1-dimensional subspace in $\mathbb{C}^2$ and so the same orthogonal projection operator $\Pi$.  
\end{rem}
We shall now give the spectra of all these self-adjoint extensions and the their corresponding Green's functions. By analysing the spectra or by the finiteness of the deficiency indices and using Theorem~\ref{glazman}, all s.a.e.s are bounded-below. These results can either be reached via the approach of Posilicano~\cite{v} (Example 5.1) or the methods of Stakgold~\cite{y}. The proofs of all but the case of the Dirichlet extension are found in Section~H of Bullock~\cite{me}. The case of the Dirichlet extension itself is simpler, along similar lines and is to be found in Stakgold~\cite{y}. The numbering below corresponds to the three groupings of s.a.e.s previously introduced: s.a.e.s of zeroth, first and second kinds. (We denote $\mathbb{N}_0=\mathbb{N}\cup\{0\}$.)

\begin{enumerate}
\setcounter{enumi}{-1}
\item We have $\sigma(A_D)=\left\{\left(\frac{n\pi}{a}\right)^2,\; n\in \mathbb{N}\right\}.$

\item If $\lambda\neq 0$ then $\lambda\in\sigma(A_\theta)$ iff
\begin{align*}
0=\;&\theta_{11}\sqrt{\lambda}\cos\sqrt{\lambda}a+\theta_{22}\sqrt{\lambda}\cos\sqrt{\lambda}a-\lambda\sin\sqrt{\lambda}a\\ 
&+\theta_{11}\theta_{22}\sin\sqrt{\lambda}a-|\theta_{12}|^2\sin\sqrt{\lambda}a+2\Re(\theta_{12})\sqrt{\lambda}
\end{align*}
(note that the validity of this condition is independent of which square root of $\lambda$ we take) and: $$0\in \sigma(A_\theta)\text{ iff } 
a|\theta_{12}|^2-\theta_{11}-a\theta_{11}\theta_{22}-\theta_{22}-2\Re(\theta_{12})=0.$$
For instance, letting $\theta_{11}=\theta_{22}=\theta_{12}=0$, we have  $\sigma(A_N)=\left\{\left(\frac{n\pi}{a}\right)^2,\; n\in \mathbb{N}_0\right\}.$
\item If $\lambda\neq 0$ then $\lambda\in\sigma(A_{w_1w_2\theta})$ iff $$-\sqrt{\lambda}\cos\sqrt{\lambda}a+2\Re (w_1\overline{w_2})\sqrt{\lambda}-\theta\sin\sqrt{\lambda}a=0$$
and $0\in \sigma(A_{w_1w_2\theta})$ iff $a\theta-2\Re(w_1\overline{w_2})+1=0$.
\end{enumerate}

The Green's functions for each of these cases is treated in the following:
\begin{enumerate}
\setcounter{enumi}{-1}
\item The Green's function for the Dirichlet extension $A_D$ is given in terms of the kernel $g(x,y;\lambda)$. For $\lambda\in \mathbb{C}\backslash \{(\frac{n\pi}{a})^2, n\in \mathbb{N}_0\}$, define:
$$g(x,y;\lambda)=\frac{\sin\sqrt{\lambda}(a-x_>)\sin\sqrt{\lambda}x_<}{\sqrt{\lambda}\sin a\sqrt{\lambda}},$$
where $x_<\colon=\min\{x,y\}$, $x_>\colon=\max\{x,y\}$. And for $\lambda=0$, let:
$$g(x,y;0)=\frac{(a-x_>)x_<}{a}.$$

\item The Green's function for the s.a.e.\ of the first kind is given by the following. For $\lambda\in\rho(A_\theta)\backslash\{0\}$:
$$\begin{array}{l}
g(x,y;\lambda)\\
=A\left[\begin{array}{l}
\lambda \cos\sqrt \lambda(a-x_>)\cos \sqrt\lambda x_<+\theta_{22}\sqrt\lambda\sin \sqrt\lambda(a-x_>)\cos\sqrt\lambda x_<\\
+\theta_{11}\sqrt\lambda \cos \sqrt\lambda(a-x_>)\sin \sqrt\lambda x_< +\theta_{11}\theta_{22}\sin \sqrt\lambda(a-x_>)\sin \sqrt\lambda x_<\\
+|\theta_{12}|^2\sin \sqrt\lambda (x_>-a)\sin \sqrt\lambda x_< + C(x,y)(\theta_{12})\sqrt\lambda \sin\sqrt\lambda(x_<-x_>)\end{array}\right],
\end{array}$$
where
$$A^{-1}=\sqrt{\lambda}\left[\begin{array}{l}
\theta_{11}\sqrt{\lambda}\cos\sqrt{\lambda}a+\theta_{22}\sqrt{\lambda}\cos\sqrt{\lambda}a-\lambda\sin\sqrt{\lambda}a\\ 
+\theta_{11}\theta_{22}\sin\sqrt{\lambda}a-|\theta_{12}|^2\sin\sqrt{\lambda}a+2\Re(\theta_{12})\sqrt{\lambda}\end{array}\right],$$
and for $k\in\mathbb{C}$, 
$$C(x,y)(k)=\left\{\begin{array}{l}
k \text{ if } x<y.\\
\overline{k} \text{ if } x\geq y. \end{array}\right.$$

If $0\in \rho(A_\theta)$, then
$$g(x,y,0)=A\left[\begin{array}{l}
(a-x_>)x_<|\theta_{12}|^2-\theta_{11}x_<+(x_>-a)x_<\theta_{11}\theta_{22}\\
+(x_>-a)\theta_{22}-1+C(x,y)(\theta_{12})(x_>-x_<)\end{array}\right],$$
where $A^{-1}=a|\theta_{12}|^2-\theta_{11}-a\theta_{11}\theta_{22}-\theta_{22}-2\Re(\theta_{12}).$

In particular, the Green's function for the Neumann extension $A_N$ is given as follows:

For $\lambda\in \mathbb{C}\backslash \{(\frac{n\pi}{a})^2, n\in \mathbb{N}_0\}$, define:
$$g(x,y;\lambda)=-\frac{\cos\sqrt{\lambda}(a-x_>)\cos\sqrt{\lambda}x_<}{\sqrt{\lambda}\sin a\sqrt{\lambda}}.$$

\item The Green's function for the s.a.e.\ of the second kind is given as follows. 

For $\lambda\in\rho(A_{w_1w_2\theta})\backslash \{0\}$,
$$\begin{array}{l}
g(x,y;\lambda)\\
=A\left[\begin{array}{l}
|w_1|^2\sqrt{\lambda}\sin\sqrt{\lambda}(x_>-a)\cos\sqrt{\lambda}x_<+\sqrt{\lambda}C(x,y)(w_1\overline{w_2})\sin\sqrt{\lambda}(x_<-x_>)\\
+\theta \sin\sqrt{\lambda}(x_>-a)\sin\sqrt{\lambda}x_<-|w_2|^2\sqrt{\lambda}\cos\sqrt{\lambda}(x_>-a)\sin\sqrt{\lambda}x_<\end{array}\right],
\end{array}$$
where
$A^{-1}=\sqrt{\lambda}\left[-\sqrt{\lambda}\cos\sqrt{\lambda}a+2\Re (w_1\overline{w_2})\sqrt{\lambda}-\theta\sin\sqrt{\lambda}a\right].$

If $0\in \rho(A_{w_1w_2\theta})$, then
$$g(x,y;0)=A[\theta(a-x_>)x_<+C(x,y)(w_1\overline{w_2})(x_>-x_<)+|w_1|^2(a-x_>)+|w_2|^2x_<],$$
where $A^{-1}=a\theta-2\Re(w_1\overline{w_2})+1.$
\end{enumerate}

\subsection{Self-adjoint Extensions of minus the Laplacian plus mass}\label{plusmass}
We shall show here that the s.a.e.s of the operator $A=-\divergence_h\circ\grad_h+m^2$ is easily given in terms of the s.a.e.s of $-\divergence_h\circ\grad_h$. The corresponding Green's functions are then easily constructible. This situation is covered by the following more general problem:
\begin{prop}\label{plusmasses}
Let $A$ be a linear operator on a (real or complex) Hilbert space $H$. For $\mu\in \mathbb{R}$ define the linear operator $A+\mu$ via the domain $D(A+\mu)=D(A)$, $(A+\mu)\phi=A\phi+\mu\phi$. Then the following are true:
\begin{enumerate}
\item $A$ is closable iff $A+\mu$ is closable.
\item If $A$ is closable then: $\overline{A+\mu}=\overline{A}+\mu$.
\item $A$ is self-adjoint iff $A +\mu$ self-adjoint.
\item $A$ is e.s.a. iff $A+\mu$ is e.s.a..
\item If $A$ is a symmetric linear operator and $\{A_\gamma\colon\;\gamma\in \Gamma\}$ are all the s.a.e.s of $A$ ($\Gamma=\phi$ is possible). Then the s.a.e.s of $A+\mu$ are precisely $\{A_\gamma+\mu\colon\;\gamma\in \Gamma\}$. Additionally, $\sigma(A_\gamma+\mu)=\sigma(A_\gamma)+\mu$ and if $G_\lambda$ is the resolvent of $A_\gamma$ at $\lambda\in \rho(A_\gamma)$ then $G_{\lambda+\mu}$ is the resolvent of $A_\gamma +\mu$ for $\lambda+\mu\in\rho(A_\gamma+\mu)$.
\end{enumerate}
\end{prop}
\begin{proof}
The proposition is easily proven directly by the definitions of closability, self-adjointness etc. 
\end{proof}
We now apply this proposition to the problem of finding the s.a.e.s of the Klein-Gordon operator on a Riemannian manifold, for which we already know all the s.a.e.s of minus the Laplacian. For instance let $\Sigma=S^1$. This was treated in Section~\ref{circle}. Let $H=L^2(S^1)$ (the Borel measure on $S^1$ being induced by the Riemannian metric on $S^1$). Define the linear operator $A$ on $H$:
\begin{align*}
D(A)&=[C_0^\infty(S^1)]=[C^\infty(S^1)]\\
A([\phi])&=-[\phi'']\text{ for } \phi\in C^\infty(S^1).
\end{align*}

Consider $B=A+m^2$. We are using the previous notation. So $D(B)=D(A)$. Then, according to the previous proposition, as $A$ is e.s.a. so also is $B$ and:
$$D(\overline{B})=D(\overline{A})=W^{2,2}(S^1)=\{\phi\in L^2(S^1) \text{ s.t. } \phi',\;\phi''\in L^2(S^1)\}.$$
Since $\sigma(\overline{A})=\sigma_{disc}(\overline{A})=\{n^2\colon \;n\in \mathbb{N}_0\}$, then the spectrum of $\overline{B}$ is:
$$\sigma(\overline{B})=\sigma_{disc}(\overline{B})=\{n^2+m^2\colon \;n\in \mathbb{N}_0\}.$$
Using the chart: $\phi\colon\; U=S^1\backslash \{1\}\rightarrow (0,2\pi)$, $\phi^{-1}(\theta)=\exp i\theta$,
define the function $g\colon\; U\times U\times \mathbb{C}\backslash\{n^2+m^2\colon\; n\in\mathbb{N}_0\}\rightarrow \mathbb{C}$ by:
$$g(\theta, \phi; \lambda)=\frac{i}{2\sqrt{\lambda-m^2}}\left[ \exp i\sqrt{\lambda-m^2}|\theta-\phi|+\frac{2\cos \sqrt{\lambda-m^2}(\theta-\phi)}{\exp (-2\pi i\sqrt{\lambda-m^2})-1}\right].$$

Clearly as before, this expression for $g$ does not depend on the choice of square root of $\lambda-m^2$ taken. It follows from Proposition~\ref{plusmasses} that $g$ so defined is the Green's function for $\overline{B}$, that is, it generates its resolvent. 

\subsection{Example of Wald Dynamics satisfying $\supp\phi\nsubseteq J(K)$}\label{counterexample}
We shall show here by means of the examples just given that there exist simple standard static spacetimes and Wald  dynamics generated by a s.a.e.\ $A_E$ such that $\supp\phi\nsubseteq J(K)$ for some initial data $(\phi_0,\dot{\phi}_0)$, where $K=\supp\phi_0\cup \supp\dot{\phi}_0$.

Consider the example considered in Section~\ref{interval}, that is $\Sigma=(0,a)$, so $M=\mathbb{R}\times (0,a)$, $g=dt^2-dx^2$. (See Figure~\ref{counter}). Take, for instance, 
$\theta=${\tiny$\begin{pmatrix}
0 & 1 \\
1 & 0
\end{pmatrix}$}
and pick the s.a.e.\ $A_\theta$ of $A$. Thus its domain is given by: $D(A_\theta)=\{\phi\in W^{2,2}(0,a)\colon\;\phi'(0)=\phi(a),\; \phi(0)=-\phi'(a)\}.$
From this we can see that if $\phi'(t, a)\neq 0$ then so is $\phi(t,0)$ and hence $\phi$ is non-zero in a neighbourhood of $(t,0)$ and so also non-zero at points outside $J(K)$.

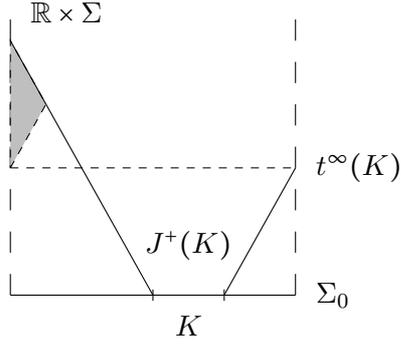
\begin{SCfigure} \centering
\begin{tikzpicture}[scale=0.75] 
\draw[dash pattern=on 8pt off 8pt] (0,0) -- (0,5) node[label=right:$\mathbb{R}\times\Sigma$]{};
\draw[dash pattern=on 8pt off 8pt] (5,0) -- (5,5);
\draw(5,0) node[label=right:$\Sigma_0$]{};
\draw(5,2.25) node[label=right:$t^\infty(K)$]{};

\draw(3.125,0) node[label=below:$K$]{};
\draw(0,0)--(5,0);

\draw(3.125,1.5) node[label=below:$J^+(K)$]{};
\draw (2.5,-0.1)--(2.5,0.1);
\draw (3.75,-0.1)--(3.75,0.1);

\draw (3.75,0)--(5,2.25);
\draw[dashed] (0,2.25)--(5,2.25);
\draw[black,dashed,fill=gray!50](0,2.25)--(0.625,3.375)--(0,4.5)--cycle;
\draw (2.5,0)--(0,4.5);
\end{tikzpicture}
\caption{Wald dynamics satisfying: $\supp\phi\nsubseteq J(K)$, where $(M,g)=(\mathbb{R}\times (0,1),dt^2-dx^2)$.
For some s.a.e.s $A_E$ there exist points in the shaded area at which $\phi$ is non-zero though they are clearly not contained in $J^+(K)$.\hspace{2em}}
\label{counter}
\end{SCfigure}

\subsection{Example of non-bounded below acceptable self-adjoint extensions of minus the Laplacian}\label{acceptablenonbounded}
We shall, in this section, construct examples of non-bounded below acceptable s.a.e.s of minus the Laplacian on certain choices of simple Riemannian manifolds (though our example shall be on a disconnected manifold). In particular, this shows that the class of solutions to the Cauchy problem of the Klein-Gordon equation on non-globally hyperbolic spacetimes constructed in this paper a nontrivial extension of the application of theory of Wald \cite{b} from bounded-below s.a.e.s to acceptable s.a.e.s (Wald considered only those that were positive).

Before we begin the construction, we shall briefly describe some necessary background. It concerns the the direct sum of linear operators on Hilbert spaces. Given a sequence $H_n$ of (real or complex) Hilbert spaces, then the direct sum is defined as usual as: $$H=\bigoplus_{n\in\mathbb{N}}H_n:=\left\{(\phi_n)_n\in \mathbb{N}\text{ such that }\sum_{n\in\mathbb{N}}||\phi_n||^2_n<\infty\right \},$$
where $||\cdot||_n$ is the norm in the Hilbert space $H_n$. 

\begin{defn}
For each $n\in\mathbb{N}$, let $A_n$ be a linear operator on the Hilbert space $H_n$. Then define the linear operator $A$ on $H$ as the direct sum of the linear operators $A_n$ as follows: 
\begin{align*}
D(A)&=\left \{\phi=(\phi_n)_{n\in\mathbb{N}}\text{ such that } \phi_n\in D(A_n)\text{ and } \sum_n||A_n\phi_n||_n^2<\infty\right\} \\
(A\phi)_n&=A_n\phi_n.
\end{align*}
We then define the countable direct sum $\bigoplus_{n\in\mathbb{N}}A_n$ of the operators $A_n$ to be the operator $A$.
\end{defn}
\begin{prop}\label{directsumops}
The following are true:
\begin{enumerate}
\item If all the linear operators $A_n$ are densely defined (closed, symmetric, self-adjoint), then $A$ is densely-defined (closed, symmetric, self-adjoint) respectively. 
\item If all the linear operators $A_n$ are bounded, then the sequence $(||A_n||_n)$ is bounded iff $A=\bigoplus_{n\in\mathbb{N}}A_n$ is a bounded linear operator.
\item The spectrum $\sigma(A)$ is obtained from the spectra $\sigma(A_n)$ by the relation: $\sigma(A)=\bigcup_n\sigma(A_n)$.
\item If each operator $A_n$ is a orthogonal projection operator on $H_n$, then $A$ is an orthogonal projection operator on $H$.
\item If all the operators $P_n$ are projection-valued measures (p.v.m.s) on $H_n$, then $P$ is a p.v.m.\ on $H$ defined by: $$P_\Omega=\bigoplus_{n\in\mathbb{N}}(P_n)_\Omega \text{ for each }\Omega\subseteq \mathbb{R} \text{ Borel.}$$
We shall denote this p.v.m.\ $P$ by $\bigoplus_{n\in\mathbb{N}}P_n$.
\item Let all the operators $A_n$ be self-adjoint. If $P_n$ is the projection-valued measure (p.v.m.) on $H_n$ associated to $A_n$ via the spectral theorem and $P$ is the p.v.m.\ on $H$ associated to $A$, then:$$P=\bigoplus_{n\in\mathbb{N}}P_n.$$
\item If all the operators $A_n$ are self-adjoint and $f:\mathbb{R}\rightarrow \mathbb{R}$ is Borel measurable, then: $$f(\bigoplus_{n\in\mathbb{N}}A_n)=\bigoplus_{n\in\mathbb{N}}f(A_n),$$
where we are using the spectral theorem to define the self-adjoint operators $f(\bigoplus_{n\in\mathbb{N}}A_n)$ and $f(A_n)$.
\end{enumerate}
\end{prop}
The proof of this proposition is an exercise in functional analysis (see e.g.\ Reed and Simon~\cite{f} or Birman and Solomjak~\cite{birman}). The proof is omitted here for brevity.

Using this notation, we construct such extensions as follows: Given a fixed Riemannian manifold $(\Sigma,h)$, we shall first consider the case of constructing a s.a.e.\ $A$ of minus the Laplacian on $(\Sigma',h)=(\mathbb{Z}\times \Sigma,h)$ from s.a.e.s $(A_n)_{n\in\mathbb{Z}}$ of minus the Laplacian on $(\Sigma,h)$. We then show that if all the s.a.e.s $A_n$ are acceptable, then so is $A$. We then give necessary and sufficient conditions for $A$ to be non-bounded below before giving a concrete example. We state our results in the form of the following proposition.

\begin{prop}\label{nonboundedbelow}
Fix a Riemannian manifold $(\Sigma,h)$ and define $(\Sigma',h)=(\mathbb{Z}\times \Sigma,h)$. Considering the Hilbert spaces $L^2(\Sigma, d\vol_h)$ and $L^2(\Sigma', d\vol_h)$ then we have the following isomorphism:
$$L^2(\Sigma', d\vol_h)\cong \bigoplus_{n\in\mathbb{Z}}L^2(\Sigma, d\vol_h).$$
Define the following linear operators $A$ and $A'$ on the Hilbert spaces $L^2(\Sigma, d\vol_h)$ and\linebreak $L^2(\Sigma', d\vol_h)$ as follows:
$D(A)=[C_0^\infty(\Sigma)]$, $A[\phi]=[-\largesquare_h\phi]$ for $\phi\in C_0^\infty(\Sigma)$ and similarly for $A'$. For simplicity, we shall treat the aforementioned isomorphism as an identification. Then we have the following relationship between $D(A)$ and $D(A')$:
$$D(A')=\left\{\phi\in \bigoplus_{n\in\mathbb{Z}}D(A)\text{ such that }\phi_n\neq 0 \text{ for at most finitely many $n$}\right\}.$$
Now, for each $n\in\mathbb{Z}$, let $A_{E,n}$ be a s.a.e.\ of $A$ and define the operator $A'_E=\bigoplus_{n\in\mathbb{Z}}A_{E,n}$. Then:

\begin{enumerate}
\item $A'_E$ is a s.a.e.\ of $A'$.
\item If for all $n$, $A_{E,n}$ is an acceptable s.a.e.\ of $A$, then $A'_E$ is an acceptable s.a.e.\ of $A'$.   
\item $\sigma(A'_E)=\bigcup_{n\in\mathbb{Z}}\sigma(A_{E,n})$.
\end{enumerate}
Therefore, if for all $n$, $A_{E,n}$ is an acceptable s.a.e.\ of $A$ and if $\bigcup_{n\in\mathbb{Z}}\sigma(A_{E,n})$ has no lower bound in $\mathbb{R}$, then $A'_E$ is a non-bounded below acceptable s.a.e.\ of $A'$.
\end{prop} 

\begin{rem}
In the equation relating $D(A)$ with $D(A')$ we are adopting the following notation: If $H=\bigoplus_{n\in \mathbb{Z}}H_n$ is the countable direct sum of Hilbert spaces, then if for each $n$, $V_n\leq H_n$ is a (not necessarily closed) subspace, then we define: $$\bigoplus_{n\in \mathbb{Z}}V_n:=\left\{\phi\in \bigoplus_{n\in \mathbb{Z}}H_n \text{ such that } \phi\in V_n \text{ for each }n\right\}. $$ With this notation, note that in general: $D(\bigoplus_{n\in \mathbb{Z}}A_n)\neq\bigoplus_{n\in \mathbb{Z}}D(A_n)$ but rather: $$D(\bigoplus_{n\in \mathbb{Z}}A_n)=\left\{\phi\in \bigoplus_{n\in \mathbb{Z}}D(A_n)\text{ such that } \sum_{n\in\mathbb{Z}}||A_n\phi_n||^2<\infty\right\}.$$
\end{rem}
\begin{proof}[Proposition~\ref{nonboundedbelow}]
It follows from the previous proposition, that $A_E$ is a self-adjoint operator on $L^2(\Sigma',d\vol_h)$. We must first show that it is in fact a self-adjoint extension of $A'$. 

By definition of $A'_E$ we have:
\begin{align*}
D(A')&=\left\{\phi\in \bigoplus_{n\in\mathbb{Z}}D(A)\text{ such that }\phi_n\neq 0 \text{ for at most finitely many $n$}\right\}\\
&\subseteq\left\{\phi\in \bigoplus_{n\in\mathbb{Z}}D(A_{E,n})\text{ such that }\phi_n\neq 0 \text{ for at most finitely many $n$}\right\}\\
&\subseteq\left \{\phi\in \bigoplus_{n\in\mathbb{Z}}D(A_{E,n})\text{ such that } \sum_n||A_{E,n}\phi_n||_n^2<\infty\right\} \\
&=D(A'_E)
\end{align*}
If $\phi\in D(A')=\left\{\phi\in \bigoplus_{n\in\mathbb{Z}}D(A)\text{ such that }\phi_n\neq 0 \text{ for at most finitely many $n$}\right\},$
then $(A_E\phi)_n=A_{E,n}\phi_n=A\phi_n=(A'\phi)_n$. This proves Statement~\textit{1}.

Let $A_{E,n}$ be an acceptable s.a.e.\ of $A$ for each $n$. Then, for all $t>0$:
\begin{align*}
[C_0^\infty(\Sigma')]&=D(A')\\&=\left\{\phi\in \bigoplus_{n\in\mathbb{Z}}D(A)\text{ such that }\phi_n\neq 0 \text{ for at most finitely many $n$}\right\}\\
&\subseteq \left\{\begin{array}{l}\phi\in \bigoplus_{n\in\mathbb{Z}} D(\exp((A^-_{E,n})^{1/2}t))\\\text{and }\sum_{n\in\mathbb{Z}}||\exp((A^-_{E,n})^{1/2}t) \phi_n||^2<\infty\end{array}\right\}\\
&=D(\exp((A'^-_E)^{1/2}t)),
\end{align*}
where the last equality follows from Statement~\textit{7} of Proposition~\ref{directsumops}, which gives:\linebreak $\exp((A'^-_E)^{1/2}t)=\bigoplus_{n\in\mathbb{Z}}\exp((A^-_{E,n})^{1/2}t)$. Therefore, $A'_E$ is also an acceptable s.a.e.\ of $A'$. 

The last statement follows from Statement~\textit{3} of the previous proposition. 
\end{proof}
\begin{lem}
Let $\Sigma=(0,\infty)$ and pick a sequence $\alpha_n\in (-\frac{\pi}{2},\frac{\pi}{2}]$ indexed by $n\in\mathbb{Z}$ such that for all $0<\epsilon<\frac{\pi}{2}$ there exists $n$ with $\alpha_n\in(-\epsilon,0)$. Define $A_{E,n}=A_{\alpha_n}$. Then, the corresponding operator $A'_E=\bigoplus_{n\in\mathbb{Z}}A_{\alpha_n}$ is a non-bounded below acceptable s.a.e.\ of minus the Laplacian. 
\end{lem}
\begin{proof}
Note that the s.a.e.s $A_\alpha$ were defined in Section~\ref{halfline} as: 
$$D(A_\alpha)=\{\phi\in W^{2,2}(0,\infty)\text{ s.t. } \cos \alpha \;\rho(\phi)=\sin \alpha \;\tau(\phi)\}.$$
The spectra of these s.a.e.s were given by the following: 
$$\sigma(A_\alpha)=\left\{\begin{array}{ll}[0,\infty)&\text{ for } \alpha\in[0,\pi/2]\\
\left[0,\infty)\cup\{-\cot^2\alpha\}\right.&\text{ for } \alpha\in(-\pi/2,0)\end{array}\right..$$
So, since $\lim_{x\rightarrow 0}\cot^2 x=\infty$ then, by Statement~\textit{3} of Proposition~\ref{nonboundedbelow}, $\sigma(A'_E)=\bigcup_{n\in\mathbb{Z}}\sigma(A_{\alpha_n})$ has no lower bound in $\mathbb{R}$, i.e.\ $\inf\sigma(A'_E)=-\infty$ and $A'_E$ is not bounded below. That $A'_E$ is still an acceptable s.a.e. of minus the Laplacian follows from the previous proposition.  
\end{proof}
\section{Summary}

We have shown the existence and uniqueness properties of solutions of the Klein-Gordon equation on arbitrary standard static spacetimes based on ``acceptable" self-adjoint extensions $A_E$ of the symmetric linear operator $A$, as defined in equation~\eqref{operatorA}. The proof of the existence (Section~\ref{existence1}) was based on work by Wald $\cite{b}$, though differs in the following: Our treatment utilised the more recent result of Bernal and Sanchez $\cite{h}$. Also, we have shown that the construction of solutions is valid also when the self-adjoint extension is merely acceptable (Definition~\ref{acceptablesae}). 

Separate to the work of Wald, we prove in this paper a result concerning the uniqueness of the Wald solutions and used this to prove a result on their support. The stronger statement: $\supp \phi\subseteq J(K)$ for $K=\supp(\phi_0)\cup\supp(\dot{\phi}_0)$, which was a condition on the dynamics in the later paper by Wald and Ishibashi on this topic $\cite{c}$, was seen to be false in general. In Section~\ref{counterexample} we give a simple example where $\supp \phi\nsubseteq J(K)$.

Also, using the uniqueness result, we defined the space of solutions in Definition~\ref{spaceofsolutions}, constructed both the ``energy form" and the ``symplectic form" on the space of solutions (Sections~\ref{sec:bilinearform} and~\ref{sec:sympform} respectively) and analysed some symmetries of the space of solutions (Section~\ref{symmetries}). 

In Section~\ref{acceptablenonbounded} we constructed an acceptable non-bounded below s.a.e.\ $A_E$ of minus the Laplacian on a particular (disconnected) Riemannian manifold (specifically: $\Sigma=\mathbb{Z}\times(0,\infty)$ with the Riemannian metric induced from that of $\mathbb{R}^2$). This example then shows that the extension of the theory of Wald~\cite{b} from bounded-below s.a.e.s to acceptable s.a.e.s carried out in this paper is non-trivial (Wald considered only positive s.a.e.s).

We shall now discuss avenues of further work on the subject of this paper. We list them as follows, some of which are related:
\begin{enumerate}
\item The well-posedness of the Cauchy problem for the Klein-Gordon equation often has a stronger meaning than that used in this paper. The stronger sense includes continuity of the map $C_0^\infty(\Sigma)\times C_0^\infty(\Sigma)\rightarrow C^\infty(M)$, $(\phi_0,\dot{\phi}_0)\mapsto \phi$. A problem unanswered in this paper is whether our solution to the Cauchy problem generated by an acceptable s.a.e. is well-posed in this sense. 
\item Once the answer to the previous open problem is known, a natural question in line with the paper by Wald and Ishibashi~\cite{c} is whether there are necessary and sufficient conditions on a solution to the Cauchy problem to be generated by an acceptable s.a.e.\ via this paper. Since their paper dealt with the case of sufficient conditions for the Cauchy problem to be generated by a positive s.a.e.\, then this would be an extension of their work to the present case.
\item An important question, connected with Statement~3, is whether or not there exists dynamics conserving the symplectic form constructed in Section~\ref{sec:sympform} (but possibly not conserving an energy form), that is \textbf{not} generated by a s.a.e.\ via the construction in this paper. This question posed by Kay and Studer~\cite{kay1} (Appendix A.2) is still unanswered. 
\end{enumerate}

\section*{Acknowledgments}

This work was supported by a research studentship from the Science and Technology Facilities Council held at the University of York.

\end{document}